\documentclass[english,CJK]{article}
\usepackage{geometry}
\usepackage{amsmath}
\usepackage{amssymb}
\usepackage{amsthm}

\usepackage{graphicx} 
\usepackage{subfigure}
\usepackage{float}
\usepackage{changepage}

\usepackage{multicol}
\usepackage{multirow}
\usepackage{epstopdf}
\usepackage[square,sort,comma,numbers]{natbib}

\newtheorem{assumption}{Assumption}
\newtheorem{proposition}{Proposition}
\newtheorem{corollary}{Corollary}
\newtheorem{lemma}{Lemma}

\makeatletter
\renewcommand\paragraph{\@startsection{paragraph}{4}{\z@}%
            {-2.5ex\@plus -1ex \@minus -.25ex}%
            {1.25ex \@plus .25ex}%
            {\normalfont\normalsize\bfseries}}
\makeatother
\begin{document}

\title{An Inattention Model for Traveler Behavior with e-Coupons}

\author{Han Qiu}

\maketitle

\begin{abstract}
    In this study, we consider traveler coupon redemption behavior from the perspective of an urban mobility service. Assuming traveler behavior is in accordance with the principle of utility maximization, we first formulate a baseline dynamical model for traveler's expected future trip sequence under the framework of Markov decision processes and from which we derive approximations of the optimal coupon redemption policy. However, we find that this baseline model cannot explain perfectly observed coupon redemption behavior of traveler for a car-sharing service.

    To resolve this deviation from utility-maximizing behavior, we suggest a hypothesis that travelers may not be aware of all coupons available to them. Based on this hypothesis, we formulate an inattention model on unawareness, which is complementary to the existing models of inattention, and incorporate it into the baseline model.
    Estimation results show that the proposed model better explains the coupon redemption dataset than the baseline model. We also conduct a simulation experiment to quantify the negative impact of unawareness on coupons' promotional effects. These results can be used by mobility service operators to design effective coupon distribution schemes in practice.
\end{abstract}


\section{Introduction}
\label{sec:intro}


Recently, incentive-based demand management methods have become popular in urban mobility businesses. The most well-known approach is dynamic pricing. On one hand, service operators can mitigate the shortage of supply in real time with surge pricing \cite{chen2016dynamic}, which is first deployed by Uber on a large scale and now commonly used by ride-sharing platforms. On the other hand, operators can use real-time price discount to attract users and compete with alternative mobility services.
Another commonly used tool is the coupons, which are vouchers that guarantee rights of fare reductions in consumption of products or services. Because coupons can change the fare only in the negative direction, their usage is usually limited to promotions. However, coupons are able to inform travelers about the possibility of fare reduction before they submit trip requests, while a traveler can be aware of any real-time price discount only after the service operator provides an offer. Therefore, coupons are a good complement to pricing for urban mobility service operators.

Nevertheless, the impact of coupons on traveler behavior is not immediately obvious. Unlike pricing, coupons usually have long life cycles and can have complicated redemption rules. A traveler may then develop sophisticated coupon redemption policy to optimize her aggregate utility from trips. For example, when the fare of the current trip is much lower than the face value of a coupon, a traveler will defer the redemption of this coupon to a later trip. The situation becomes increasingly complex when the traveler is presented with various coupons and some coupons may have their face values, expiration dates, and redemption rules different from those of others. For example, the traveler can have one coupon that can reduce the fare of one trip for up to 5 dollars and another to reduce the fare of one trip by 20\%.

Moreover, in a typical urban mobility service setting, coupons are involved in several different decision problems such as the travel mode selection and the coupon selection for payment. Because the service operator does not have complete information about the traveler, structural estimation of coupon impact on some decision problems is not practical. For instance, operator of one mobility service is not able to distinguish between the promotional effect in market share from her own coupon distribution strategy and the one from coupon distribution strategies of alternative mobility services. In this case, the operator can only assess the impact of coupons on the travel mode selection behavior with direct experiments on coupon distribution strategies and the result is generally not stable with respect to changes in operations of alternative mobility services.

Recognizing these difficulties, in this study, we focus on travelers' coupon selection behavior. Coupon selection is a relatively simple decision problem compared with others because the decision mostly depends on the traveler's evaluation of coupons. Moreover, for this problem, we are able to acquire adequate data for model estimation.

Specifically, in this study, we consider a setting in which an operator provides an urban mobility service via a mobile app and occasionally distributes electronic coupons (e-coupons) to attract travelers and to boost revenue. Moreover, to focus on the traveler behavior, we assume that travelers' decisions are independent of the operator's coupon distribution strategy. The e-coupon is in a simple format: each e-coupon can be represented by its face value and expiration date. These e-coupons are stored in an online electronic wallet (e-wallet) within the app and visually accessible only when travelers open their e-wallet. After a trip, the traveler can select up to one coupon for redemption, and a coupon can reduce the payment by subtracting from the trip fare its face value, at most.

In selecting a coupon for redemption, the traveler evaluates for each coupon both its immediate redemption value and the future redemption value of the rest of coupons. To infer these evaluations, we first notice that the immediate redemption value can be obtained directly from the definition of a coupon. Then, we develop a dynamical model of trip sequences to estimate the values of future coupon redemption, because these values depend on trips happened during the life cycle of the coupon set. In particular, we view the trip generation and realization process (expected by the traveler) as a Markov decision process (MDP) and derive value approximations by assuming the utility-maximization behavior of travelers and solving for the optimal value function of the MDP. The resulting optimal policy also dictates the coupon selection behavior of a utility-maximization traveler.

However, we find that the observed coupon selection behavior of travelers for a car-sharing service deviates from the above model. To explain this finding, we discuss several possibilities and finally come up with the hypothesis that travelers may not be aware of all available coupons. We then provide a mathematical formulation for such patterns of unawareness and incorporate it into the aforementioned MDP.

Subsequent estimation results show that the proposed unawareness model better explains the coupon redemption dataset than baseline models. A simulation experiment further shows that if such unawareness exists the reduction rate on coupons' promotional effects can be as great as 10\%.

The contribution of this work is twofold. First, the proposed unawareness model is complementary to the existing models of inattention and they have completely different behavioral implications. Therefore, when existing inattention models fail to explain some dataset, one can consider the unawareness model. Second, our estimation results can be used by mobility service operators in designing effective coupon distribution schemes in practice.

The rest of this paper is organized as follows. In Section \ref{sec:lr}, we summarize related works on coupons, estimation of dynamic behavior, and limited attention. In Section \ref{sec:model}, we formulate the dynamical model for the trip sequences and derive the optimal coupon redemption policy and the corresponding evaluation of available coupons. In Section \ref{sec:data}, we describe the dataset from a car-sharing service and show the discrepancies between the observed coupon redemption behavior and decisions dictated by the optimal policy. In Section \ref{sec:att_model}, we extend the baseline dynamical model with a mathematical formulation of unawareness. In Section \ref{sec:results}, we summarize the estimation results of the proposed model. In Section \ref{sec:sim_exp}, we conduct a simulation experiment to assess the impact of unawareness on the promotional effect of coupons. Finally, in Section \ref{sec:conclusion}, we conclude our work and suggest directions for future research.

In later discussions, we use the terms ``inattention'' and ``unawareness'' interchangeably. Moreover, general inattention behavior that is not explicable by unawareness is denoted as ``deliberate attention''.

\section{Literature Review}
\label{sec:lr}

\subsection{Customer behavior with coupons}

As a major type of price promotions, coupons have been studied in marketing literature for decades, resulting in numerous theoretical and empirical papers. Here, we provide only a brief summary of the major focus of this literature and highlight the novelty of our study.

Decades ago, merchants usually designed coupons with simple redemption rules. For example, a grocery store can issue coupons for a specific product with the same face value and valid period. Given this simplicity, the optimal coupon redemption strategy is trivial to compute. Therefore, research on customer behavior at that time mainly focused on customers' \textit{coupon proneness}, or the (latent) intention of obtaining and redeeming coupons. Such \textit{coupon proneness} was usually estimated from socio-economic factors \cite{lichtenstein1990distinguishing,bawa1997coupon,jayasingh2015empirical}.
Simple redemption rules also lead to aggregate coupon redemption patterns that can be described by some elementary functions. For example, Ward and Davis \cite{ward1978coupon} suggested that after coupon issuance, the coupon redemption rate declines exponentially as time passes. Inman and McAlister \cite{inman1994coupon} further considered the impacts of expiration dates and extended the above model with a hyperbolic function.

With the rising popularity of online shopping, merchants now prefer to use e-coupons over traditional paper coupons in price promotions.
This change leads to two patterns. First, because e-coupons can be distributed with low cost and in large-scale, merchants can now reach out to a large number of customers and each customer may receive coupons from a variety of merchants.
Second, because the redemption of e-coupons is processed by computers rather than human beings, merchants are now able to develop complicated redemption rules.
Under these patterns, the structural estimation of customer behavior with coupons becomes considerably more difficult.
Therefore, recent research mostly focused on reduced-form models or even data-driven approaches.
For example, Reimers and Xie \cite{reimers2018coupons} proposed reduced-form models for coupons' market expansion and revenue cannibalization effects. The authors estimated their models using restaurant coupon data from Groupon.
Zhang \textit{et al.} \cite{zhang2017does} investigated the short and long-term effects of coupon distributions on customer behavior by conducting randomized field experiments on the Taobao Marketplace, the largest online C2C platform in China. The authors applied linear models to explain their experimental findings at an aggregate level.

Our work differs from the ones mentioned above in that we consider a structural model under general coupon redemption schemes. We point out that the proposed dynamical structural model is similar to the ones used in the estimation of multi-stage household consumption \cite{gonul1996estimating,blattberg2010identifying}. We briefly discuss this connection in the next subsection.

\subsection{Inference on temporal correlated behavior}

When people model real-world human behavior, they usually assume that human behaves according to the utility-maximization principle. Then, the inference problem can be reduced to a model estimation problem for the utility function. However, when the observations come from some sequential decision processes, the corresponding utility-maximizing policy needs to consider not only the utility from immediate actions but also those in the future. In this case, the connection between the (single-step) utility function and the (multiple-step) utility-maximizing policy is not immediately obvious.

In the field of econometrics, such estimation problems are generally framed as the dynamic discrete choice (DDC) problems. The first practical estimation algorithm for DDC models, the nested fixed point algorithm, is proposed by Rust \cite{rust1987optimal} in the 1980s to describe the engine replacement behavior of a bus company. This estimation algorithm refers to a two-stage optimization process: first, find the optimal value function $V$ corresponding to the utility function determined by $\theta$; then, do local searches for a better $\theta$ according to the estimates $V$. The huge computational burden makes this algorithm intractable for more general use, and from then on more computationally efficient algorithms, such as Hotz-Miller's conditional choice probability method \cite{hotz1993conditional}, have been suggested. Interested readers are referred to Aguirregabiria and Mira \cite{aguirregabiria2010dynamic} and Heckman and Navarro \cite{heckman2007dynamic} for more details on these recent developments.

The estimation algorithms of dynamic discrete choice models share many similarities with the reward learning methods for the inverse reinforcement learning (IRL) problems. Compared with dynamic discrete choice models, these methods assume less structural knowledge of the utility (reward) function and allow for more freedom on the choice of the function form. For instance, we can use deep neural networks to capture the complex relationship between the state and the reward.
However, in this case, the estimation problem is generally ill-posed \cite{ng2000algorithms,ziebart2008maximum} and one needs to add other regularization or penalty terms to obtain meaningful reward functions. For example, Ziebart \textit{et al.} \cite{ziebart2008maximum} construct an entropy-regularized maximum likelihood estimator for IRL problems; Abeel and Ng \cite{abbeel2004apprenticeship} consider an optimization problem to find the maximum margin hyperplane that separates the expert demonstrations from other non-optimal policies.

Reward learning methods have been the major research focus of IRL problems since the last decade; however, recently, there have been more interests in methods that train policy from demonstration directly. In particular, several papers \cite{ho2016generative,fu2017learning,finn2016connection} suggest that the optimal policy corresponded to the recovered reward function from a reward learning method can be viewed as the policy learned from behavioral cloning (supervised learning) of the observed behavior under the regularization condition uniquely determined by the same method. This interpretation is then used to develop several generative-adversarial-network (GAN) based IRL methods for simultaneous policy learning and reward learning, and these methods achieve better learning performance compared with state-of-the-art baselines.

In this paper, we avoid the aforementioned difficulties in utility estimation by direct approximations. Specifically, we assume that the utility is additively separable: the total utility from each trip can be decomposed as the sum of the utility from coupon redemption, which is in the monetary unit and there is no need for estimation, and the unknown utility from the trip itself. Then, we show by mathematical derivation that the unknown utility from the trip can be safely ignored under certain regularity conditions. Under these assumptions, we can approximate the optimal value function of the sequential decision problem with a value determined only by the values of immediate or future coupon redemption, and the computation can be done at once before model estimations.

\subsection{Limited attention and consideration set}

One novelty of our work is that we explicitly model the impact of unawareness in travelers' decision dynamics.
In this subsection, we review the related literature on discrete choice problems under limited attention.

Discrete choice problems under limited attention can generally be described by a three-stage decision process \cite{shocker1991consideration}. First, an \textit{awareness subset} is drawn from the whole choice set by chance. Then the decision maker (DM) deliberately limits her attention to a \textit{consideration subset} of the \textit{awareness set}. Finally, the DM makes a choice within the \textit{consideration set}.
Most works in the literature focus only on the second stage, i.e. the generation of a \textit{consideration set}, possibly limited by available data or by the problems with identification.
Recently, theoretical works, including the one of Masatlioglu \textit{et al.} \cite{masatlioglu2016revealed}, introduced frameworks that viewed the \textit{awareness set} and the \textit{consideration set} as an individual object.
However, as illustrated later in this paper, these two terms should not be used interchangeably in a dynamic decision model because they are generated according to different mechanisms. In particular, the generation of an \textit{awareness set} is an action taken by nature and needed to be explicitly modeled with the dynamical model, whereas the generation of an \textit{consideration set} is an action taken by the DM and is implicitly modeled in a class of decision strategies(policies).
In this study, we specifically focus on the modeling of the former.

The modeling of discrete choice problems under limited attention started in the 1970s \cite{manski1977structure}. At that time, researchers were interested in theoretically attractive extensions of classic discrete choice models, e.g., the multinomial logit (MNL) model.
As an initial work, Manski \cite{manski1977structure} introduced a random set model, in which each choice is independently considered for attention. The computational complexity of this model scales with the power of the choice set size because the \textit{consideration set} can be any subset of the whole choice set. This model is referred to as the ``Manski model'' in the discussion below.
Later, Swait and Ben-Akiva \cite{swait1987empirical} developed the parameterized
logit captivity (PLC) model, in which the \textit{consideration set} can either be the whole choice set or contain a single choice option.
Empirical works \cite{ben1995discrete,swait1987empirical} showed that these models had better explanatory power than the pure MNL model.

Since the last two decades, customers have been able to browse and purchase an increasing number of products either online or via mobile apps, thanks to the developments of information technology. With very large choice sets, consumers exhibit decision patterns that deviate much from rationality but can possibly be explained by limited attention. Consequently, there had been rising interests in understanding and exploiting such behavior, and we witnessed a burgeoning number of empirical works on \textit{consideration sets}.
For example, Chiang \textit{et al.} \cite{chiang1998markov} estimated a random-parameter extension of the Manski model using data on households' choices among four ketchup brands.
Goeree \cite{goeree2008limited} estimated the Manski model using data on customer choices among personal computer products in the US.
Honka \cite{honka2014quantifying} applied the concept of searching costs to develop an attention model and estimated it using a dataset of customer choices among automobile insurance products in the US.
Honka \textit{et al.} \cite{honka2017advertising} further extended the above model by including the Manski model for the \textit{awareness set} generation. The authors estimated the resulting model using data on customer choices among bank accounts in the US.
All of these works claimed that the inclusion of the set consideration stage leads to better specifications and estimation results.

At the same time, other scholars kept progressing in the theoretical development of the limited attention mechanism.
Manzini and Mariotti \cite{manzini2014stochastic}, Masatlioglu \textit{et al.} \cite{masatlioglu2016revealed}, and Abaluck and Adams \cite{abaluck2016discrete} focused on the axiomatic formulation of the \textit{consideration set}, aiming at extending the current preference theory.
Sims \cite{sims2003implications}, Kim \textit{et al.} \cite{kim2010online}, and Gabaix \cite{gabaix2014sparsity} considered information costs in the DM's search for choice options and developed models of rational inattention.
Masatlioglu and Nakajima \cite{masatlioglu2013choice} and Seiler\cite{seiler2013impact} further extended these optimization models of option searching into dynamic decision processes.
For a comprehensive summary of recent theoretical developments, interested readers are referred to the recent work by Masatlioglu \textit{et al.} \cite{masatlioglu2016revealed}.

Restricted by application scenarios, the modeling of \textit{consideration sets} had only been applied for traveler mode choice or location choice behavior in the transportation literature. For example, Swait and Ben-Akiva \cite{swait1987empirical} estimated the PLC model using travel mode choice data in Sao-Paulo, Brazil. Ba{\c{s}}ar and Bhat \cite{bacsar2004parameterized} estimated the Manski model using data on passenger choices among four airports in the San Francisco Bay Area. Mahmoud \textit{et al.} \cite{mahmoud2016myopic} estimated the PLC model with a dataset of travel mode choices in the city of Toronto.

Our inattention model is distinct from the above models in two aspects. First, we consider behavior under limited attention in a dynamic decision process. Second, we pay more attention to the generation of \textit{awareness sets} than to the generation of \textit{consideration sets}. Simulation results in later sections show that such modeling differences actually lead to important practical implications.

\section{Baseline Model}
\label{sec:model}

In this section, we formulate a dynamical model for trip sequences to derive the utility-maximizing coupon redemption policy and the corresponding evaluation of available coupons. Without loss of generality, in following discussions we always assume that utilities are in the monetary unit.
The model formulation is decomposed into two parts. First, we specify the temporal correlation between consecutive trips. In particular, we model the trip generation process as a discrete time dynamical system with the available coupon set being the only state variable that keeps changing across trips. Secondly, we construct a structural model for individual trips. Specifically, we view an individual trip as a combination of two stages: the travel mode selection stage and the coupon selection stage. The overall structure of the dynamical model is described in Figure \ref{fig:decision_flow}. One should notice that not all elements in this figure are observable from our perspective; for instance, a trip will be observed only if the target mobility service is chosen.

\begin{figure}[!t]
\centering
  \includegraphics[width = 0.8\textwidth]{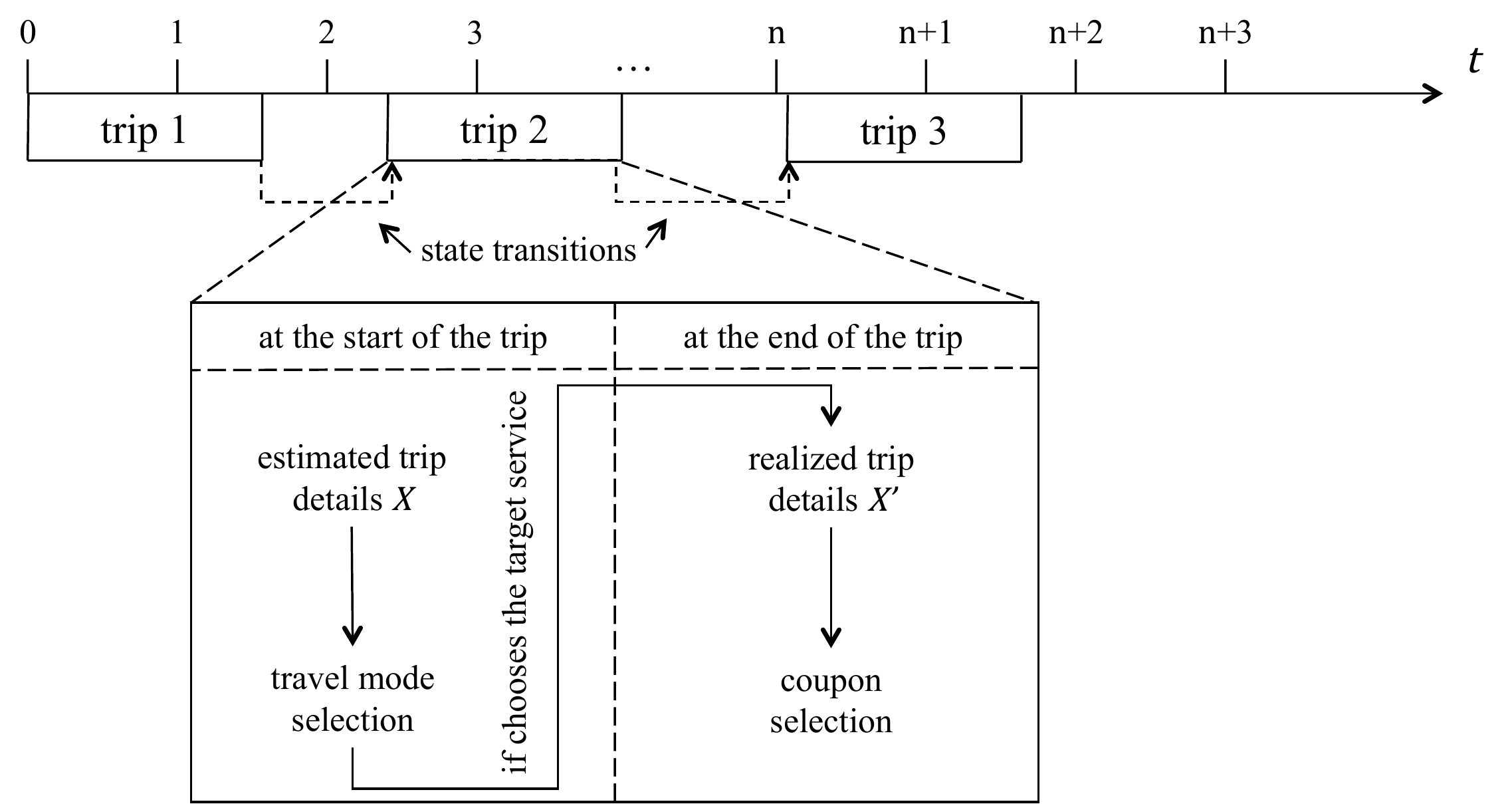}
  \caption{Structure of the dynamical model for trip sequences}\label{fig:decision_flow}
\end{figure}

Next, we summarize several general assumptions about the dynamical model that mentioned explicitly earlier in this paper.

\begin{assumption}\label{thm:agg}
    \begin{itemize}
        \item[(a)] Each coupon is represented by its face value and expiration date and can be used freely before expiration. However, for each payment at most one coupon can be redeemed, and the final trip fare must be nonnegative (that is, the reduction in fare cannot exceed the fare itself).
        \item[(b)] The traveler only considers trip events for the future; therefore, the traveler has no expectation on future coupon arrivals.
    \end{itemize}
\end{assumption}

Later in Section \ref{sec:data}, we will show that both assumptions are reasonable for our dataset.

\subsection{Notation of coupons}

Before proceeding to the model formulation, we first discuss the notations of coupons. A coupon $\tilde{c}$ consists of a face value $v$ and the remaining time to expire $T$, whereas an available coupon set includes many coupons. However, directly modeling coupon set $C$ as a set of coupons $\{\tilde{c}_1,\cdots,\tilde{c}_m\}$ does not always work here, because a mathematical ``set'' requires its every element to be distinct, but a traveler usually has several coupons with the same $v$ and $T$.

To overcome this difficulty, we model coupons in groups: $c = \langle v,T,n \rangle \in R \times N \times N^+$, where $n$ is the number of coupons in the group.
A coupon set $C$ can then be defined as a finite set of coupon groups $C = \{c_1,\cdots,c_m\}\subset R \times N \times N^+$, which is subject to the restriction that any two coupon groups $c_i,c_j$ in $C$ cannot have the same characteristics $(v,T)$.
Moreover, the coupon set $C$ should always include the option of selecting no coupon; here, this default option is represented by a coupon group of zero-valued coupons $c_0 = \langle 0,0,1 \rangle$. We further define $C_0$ as the default set $\{c_0\}$.
We use $\mathcal{C}$ to represent the set of all possible coupon sets.

Now, a natural subset $C_a$ of the coupon set $C$ is not necessarily a mathematical subset of $C$. However, given $C = \{c_1,\cdots,c_m\}$ and $c_i = \langle v_i,T_i,n_i \rangle$, we can characterize $C_a$ as follows: $C_a = \{c^a_{i}|i\in I^a\} \in \mathcal{C}$, where the index set $I^a \subset \{1,2,\cdots,m\}$, and coupon group $c^a_i = \langle v_i,T_i,n_i^a \rangle$ with $0 < n_i^a \leq n_i$.
We use $\mathcal{A}(C)$ to denote the set of all possible subsets $C_a$ of $C$.

\subsection{Temporal correlation among trips}

A simple way to describe a trip sequence is to discretize the time into steps according to a unit $t_0$ and allocate each trip to a time step. For example, we can generate the trip for traveler $j$ in a time step according to a Bernoulli distribution $B(1,\lambda_j)$ (the probability of having a trip demand in each time step is $\lambda_j$) and assuming trips from different time steps are generated independently. However, the selection of an appropriate time unit $t_0$ is not trivial in general: on one hand, when $t_0$ is large, e.g., $t_0 = 1$ week, it is unlikely that a traveler has no more than one trip within a time step; on the other hand, when $t_0$ is small, e.g., $t_0 = 1$ minute, a trip can last for many time steps and the temporal correlation between consecutive time steps can be strong. For instance, a traveler who just finished a trip 10 mins ago is less likely to submit a trip request now. Also, the computation complexity of modeling trip sequences within a specific time range is inverse proportional to $t_0$ and can lead to practical problems when $t_0$ is very small.

A more natural model of trip sequences is to consider a continuous time setting and use stochastic processes for trip generation. For example, Poisson processes generalize the independent Bernoulli sampling process in the discrete time setting. However, this model also suffers from strong temporal correlations among trips and huge computation complexity. In fact, this model is closely related to discrete time models with very small $t_0$, as shown in Appendix A.

Because the ultimate purpose in this section is to derive practical estimations of the future redemption value of coupons, here we choose to use a larger time unit $t_0$ to reduce the modeling and computation complexities. In particular, we want to select a large enough $t_0$ such that the following assumption holds:

\begin{assumption}\label{thm:a1}
    For any mode $i$, the realized trip time $t_{xi}'$ is always upper bounded by the time unit $t_0$
    \begin{equation}
        t_{xi}' \leq t_0.
    \end{equation}
\end{assumption}

This assumption says that a trip cannot last more than a time step. With this assumption, we avoid the difficulty in modeling heterogeneous state transitions across time. For urban mobility services, a $t_0$ greater than three hours is generally adequate for Assumption \ref{thm:a1} to hold. In subsequent discussion, we always assume $t_0 = 1$ day and write $t_0 = 1$ for short.

One cost of selecting a large $t_0$ is that we will underestimate the total number of trips and, therefore, the value of future coupon redemption. Luckily, as will be shown later in Section \ref{sec:data}, for our dataset most of the travelers are low-frequency users of the car-sharing service; therefore the underestimation of future coupon redemption values is not a severe problem.

Next we discuss the trip generation process in each time step. First, define $\lambda_{j,t}$ as the probability of traveler $j$ having a trip demand in time step $t$. In general, $\lambda_{j,t}$ depends on all the past information $(\lambda_{j,t-1},\cdots,\lambda_{j,0},S_{j,t-1},\cdots,S_{j,0})$, in which $S_{j,t}$ is state variables for $j$ at time $t$.
However, inclusion of these past information into the trip generation process leads to several practical issues: trips served by alternative modes are not observed, so the past information for $\lambda_{j,t}$ is always incomplete; even if we have complete information, the dependency can be highly nonlinear, which leads to a difficult estimation procedure and unstable predictions into the future.
Therefore, we make the following assumption for simplification:

\begin{assumption}\label{thm:a2}
    Trip generation rate $\lambda_{j,t}$ does not rely on any past trip of traveler $j$ and is fixed as $\lambda_{j}$.
\end{assumption}

In practice, we can let $\lambda_{j}$ be traveler $j$'s average trip generation rate. Intuitively, the estimation of coupon redemption value from this approximation is reasonable when both the real $\lambda_{j,t}$ does not differ much from $\lambda_{j}$ and the time range $T$ in consideration is large.

In summary, in this study we consider a discrete time setting for trip generations, in which the time unit $t_0$ is 1 day and the probability of having a trip demand in each time step is a traveler-specific constant $\lambda_j$.

Finally, we discuss the state transition of the available coupon set $C$ between consecutive trips. $C$ is updated in two steps: first, if there is a coupon redemption $c \neq c_0$, the used coupon $\tilde{c}$ is removed from the set $C$; secondly, every remaining coupon $\tilde{c'}$ in the set $C$ becomes one-time-unit closer to expiration. This updating procedure is described by the following state transition function $f$ and coupon group transition function $f_c$:
\begin{equation}
    \begin{aligned}
        f(C,c) & = \{f_c(c')|c' \in C/\{c\}\} \cup \{f_c(\langle v,T,n-1 \rangle)\} \\
        f_c(\langle v,T,n \rangle) & =
        \begin{cases}
            \langle v,T-1,n \rangle & v,n > 0, \ T \geq 1 \\
            c_0 & \text{otherwise}
        \end{cases}
    \end{aligned}
\end{equation}

\noindent where we assume that $c = \langle v,T,n \rangle$. Notice that in general cases the update of time to expiration $T$ depends on $t_x$; here the homogeneity is ensured by Assumption \ref{thm:a1}.

When there is no coupon selection, the state transition is described by $f(C,c_0)$. For simplicity, we use $f(C)$ to represent this default state transition $f(C,c_0)$.

\subsection{Model of an individual trip}

In the following formulation, we limit our discussion to a specific traveler $j$.

First, we simplify the interaction between the traveler and the target mobility service in an individual trip as follows: at the beginning of a trip, the traveler decides which travel mode to use and whether to cancel the trip. If the traveler selects a mode, she sticks with her choice until she arrives at the destination. Upon arrival, the traveler proceeds to payment and if she selects the target mobility service, she can decide which coupon to use. The trip ends after the payment. In other words, we decompose an individual trip into two stages: travel mode selection and coupon selection.

In the stage of travel mode selection, the detail of the trip demand $X$, including the trip distance and the current traffic situation, is revealed to the traveler. Without loss of generality, we use a traveler-specific distribution $P_j(\cdot)$ to describe the generation of $X$. Moreover, following the same logic as the discussion on $\lambda_j$ in the last subsection, we make the following assumption:

\begin{assumption}\label{thm:ag1}
    Distribution $P_j(\cdot)$ does not rely on any past trip of traveler $j$.
\end{assumption}

Given $X$, the traveler selects a travel mode or cancels the trip according to a traveler-specific policy $\pi_{xj}$ and the mode-specific information such as available coupons for different mobility services. Without loss of generality, let us assume that the potential travel modes are indexed as $0,1,\cdots, n_j$, where 0 corresponds to trip cancellation and 1 corresponds to the target mobility service. The probability of selecting mode $i$ can then be described as
\begin{equation}
    P(i) = \pi_{xj}(i|X,I_0,\cdots,I_{n_j}),
\end{equation}

\noindent where $I_i$ captures all private information about mode $i$ of traveler $j$.

Because the specific selection probabilities $P(0),P(2),\cdots,P(n_j)$ are not available to us, our discussion on the travel mode selection is restricted to the event $\{i = 1\}$ (whether the target mobility service is selected). Moreover, because we do not have the information $I_0,I_2,\cdots,I_{n_j}$, we need the following assumption for subsequent developments:

\begin{assumption}\label{thm:ag2}
    Information on the alternatives $I_0,I_2,\cdots,I_{n_j}$ are invariant among trips for any traveler $j$.
\end{assumption}

We also make an assumption to simplify the form of $I_1$:

\begin{assumption}\label{thm:ag3}
    Information on the target mobility service $I_1$ can be fully described by the available coupon set $C$; other factors, such as the service quality, are assumed to be invariant among trips for any traveler $j$.
\end{assumption}

Later in Section \ref{sec:data}, we will show that both assumptions \ref{thm:ag2} and \ref{thm:ag3} can be partly justified for our dataset.

Next, we make an assumption to simplify the form of the traveler's decision policy $\pi_{xj}$:

\begin{assumption}\label{thm:ag4}
    In selecting a travel mode, the traveler evaluates the utility $u_{ij}(X,I_i) \in R$ for each mode $i$ and her decision depends exclusively on these utilities; that is, the policy $\pi_{xj}$ has the following form
    \begin{equation}
        \pi_{xj}(i|X,I_0,\cdots,I_{n_j}) = \tilde{\pi}_{xj}(i|u_{0j}(X,I_0),\cdots, u_{n_jj}(X,I_{n_j})).
    \end{equation}
\end{assumption}

The above three assumptions immediately lead us to the form
\begin{equation}
    \pi_{xj}(i|X,C) = \tilde{\pi}_{xj}(i|u_{0j}(X),u_{1j}(X,C),\cdots, u_{n_jj}(X));
\end{equation}
\noindent to simplify notations, in following discussion we use $u_{xij}$ to denote $u_{ij}(X)$ and $\mathbf{u}_{xj}$ to denote $\{u_{x0j},\cdots,u_{xn_jj}\}$.

The next assumption is crucial for the computational tractability of our subsequent analysis.

\begin{assumption}\label{thm:ag5}
    The utility from taking the target mobility service $u_{1j}(X,C)$ is the sum of the utility from the service itself $u_{x1j}$ and the utility from potential coupon redemption $u(p_x,C)$, where $p_x \in R^+$ is the estimated trip fare with the target mobility service.
\end{assumption}

In summary, with above simplifications, the travel mode selection policy now reduces to
\begin{equation}
    P(i=1) = \pi_{xj}(p_x,\mathbf{u}_{xj},C)
\end{equation}

Next, we consider the stage of coupon selections. Assumption \ref{thm:agg}(a) restricts the action space of the traveler and the traveler's decision can now be interpreted as a probability distribution over the set of available coupons. Without loss of generality, the coupon selection probability can be expressed as $P(\tilde{c}) = \pi_{\tilde{c}}(\tilde{c}|X',C)$, where $\pi_{\tilde{c}}$ is the coupon selection policy and $X'$ captures the realized trip details. In subsequent discussion we use distribution $P_{j}'$ to describe the generation of $X'$ from $X$: $X' \sim P'_{j}(\cdot|X)$.

Because the traveler cannot distinguish among coupons in the same group $c$, we make the following assumption to simplify the form of $\pi_{\tilde{c}}$:

\begin{assumption}\label{thm:ag6}
    The traveler makes selection in two steps. First, she selects one coupon group $c$ from her available coupon set $C$ according to policy $\pi_c$: $P(c) = \pi_c(c|X',C)$. Then, she chooses a coupon $\tilde{c}$ from this group $c$.
\end{assumption}

Next we simplify the form of $X'$ in $\pi_c$. To make an optimal decision, the traveler needs to evaluate both the utility from immediate redemption and the one from future redemption. On one hand, by definition the value $r$ of redeeming a coupon in the group $\langle v,T,n \rangle$ given the realized trip fare $p_x'$ is $r(p_x',\langle v,T,n \rangle) = \min(v,p_x')$. On the other hand, because of assumptions \ref{thm:a2} and \ref{thm:ag1} the utility of future redemption does not depend on the details $X'$ of the current trip. Therefore, we can make the following assumption:

\begin{assumption}\label{thm:ag7}
    In the coupon selection stage the traveler only consider the realized trip fare $p_x'$ and the coupon set $C$ for decision; that is, $\pi_c(c|X',C) = \pi_c(c|p_x',C)$.
\end{assumption}

\subsection{Value functions and approximations}

In this subsection, we derive optimal policy and value function of the above dynamical model and develop practical approximations to characterize travelers' coupon redemption behavior.

Suppose that the mode choice policy $\pi_{xj}$ and the coupon redemption policy $\pi_c$ are given. Since our model includes several stages, to make subsequent discussions clearer, we introduce following definitions: $U^{\pi}(C)$ is the expected utility gain from the target mobility service and coupon set $C$ at the beginning of the time step and is called the ``\textit{ex ante} utility''; $U_{xj}^{\pi}(p_x,\mathbf{u}_{xj},C)$ is the expected utility gain after the revealing of trip demand details and is called the ``interim utility''; $U_c^{\pi}(p_x',C)$ is the expected utility gain at the end of trip realization stage and is called the ``\textit{ex post} utility''.

First, depending on whether there is a trip in the time step, the realized \textit{ex ante} utility $U^{\pi}(C)$ equals to either the \textit{ex ante} utility at the next time step $U^{\pi}(f(C))$, or the expected interim utility $E_{X|P_j} U_{xj}^{\pi}(p_x,\mathbf{u}_{xj},C)$:
\begin{equation}\label{eq:util_ex_ante}
        U^{\pi}(C) = (1 - \lambda_j) \gamma U^{\pi}(f(C)) + \lambda_j E_{X|P_j} U_{xj}^{\pi}(p_x,\mathbf{u}_{xj},C),
\end{equation}

\noindent where $\gamma$ is the time discount factor for a time step.

Next, depending on whether the target mobility service is selected, the realized interim utility $U_{xj}^{\pi}(p_x,\mathbf{u}_{xj},C)$ equals to either the sum of utility from alternatives and the \textit{ex ante} utility at the next time step $U^{\pi}(f(C))$, or the sum of utility from the target mobility service $u_{x1j}$ and the expected \textit{ex post} utility $E_{X'|X,P_j'}U_c^{\pi}(p_x',C)$:
\begin{equation}\label{eq:util_interim_sim}
    \begin{aligned}
        U_{xj}^{\pi}(p_x,\mathbf{u}_{xj},C) & = \sum_{i \neq 1}P(i)[u_{xij}+\gamma U^{\pi}(f(C))] + P(i=1)[u_{x1j} + E_{X'|X,P_j'}U_c^{\pi}(p_x',C)] \\
        & = (1 - \pi_{xj}(p_x,\mathbf{u}_{xj},C)) [u_{x\tilde{1}j} + \gamma U^{\pi}(f(C))] + \pi_{xj}(p_x,\mathbf{u}_{xj},C) [ u_{x1j} + E_{X'|X,P_j'}U_c^{\pi}(p_x',C)],
    \end{aligned}
\end{equation}

\noindent where the value $u_{x\tilde{1}j}$ is equal to $\frac{1}{1-P(1)}\sum_{i \neq 1} P(i) u_{xij}$ and can be interpreted as the expected utility from taking alternative mobility services. Since we do not have any specific information on each of alternative $i$, we assume that the value of $u_{x\tilde{1}j}$ is independent from the selection probability of the target mobility service $P(1)$. Now, if we define the utility gain from taking the target mobility service $u_{xj} = u_{x1j} - u_{x\tilde{1}j}$, we have

\begin{equation}\label{eq:util_interim}
    U_{xj}^{\pi}(p_x,\mathbf{u}_{xj},C) = (1 - \pi_{xj}(p_x,\mathbf{u}_{xj},C)) \gamma U^{\pi}(f(C)) + \pi_{xj}(p_x,\mathbf{u}_{xj},C) [ u_{xj} + E_{X'|X,P_j'}U_c^{\pi}(p_x',C)] + u_{x\tilde{1}j},
\end{equation}

Finally, the \textit{ex post} utility $U_c^{\pi}(p_x',C)$ equals to different \textit{ex ante} utility at the next time step, depending on which coupon the traveler selects to redeem:
\begin{equation}\label{eq:util_ex_post}
        U_c^{\pi}(p_x',C) = \sum_{c\in C}\pi_c(c|p_x',C) [r(p_x',c) + \gamma U^{\pi}(f(C,c))].
\end{equation}

For notational simplicity, in following discussion we use $E_{X}$ to replace $E_{X|P_j}$ and $E_{X'}$ to replace $E_{X'|X,P_j'}$.

We next show that the above formulation leads to a technical problem and cannot be applied directly. In fact, we can derive the following corollary by simply replacing the coupon set $C$ in equations \eqref{eq:util_ex_ante}, \eqref{eq:util_interim} and \eqref{eq:util_ex_post} with the default set $C_0$:
\begin{corollary}\label{thm:c1}
  We have
  \begin{equation}
    U^{\pi}(C_0) = \frac{1}{1 - \gamma}\lambda_j E_{X}[u_{x\tilde{1}j} + \pi_{xj}(p_x,\mathbf{u}_{xj},C_0)  u_{xj}].
  \end{equation}
\end{corollary}

The detailed proof is provided in Appendix B. This corollary says that when the time discount factor $\gamma$ is closed to 1, the \textit{ex ante} utility $U^{\pi}(C_0)$ depends critically on $\gamma$ and can be very large. However, this is often the case in the real world context: for example, if we let the yearly discount factor be 0.9, the discount factor for a day is then 0.9997. Because the utility gain contributed by coupon redemption is bounded by finite numbers, e.g., the sum of values of all coupons in the set, the estimation of coupon impacts, or the difference among $U^{\pi}(C)$ with different $C$, can be numerically unstable.

To achieve regularity in formulations and to reduce numerical instability, we simply subtract the value of $U^{\pi}(C_0)$ from every $U^{\pi}(C)$, as illustrated in the following proposition:
\begin{proposition}\label{thm:p1}
    If we define $V^{\pi}(C)$ as $U^{\pi}(C) - U^{\pi}(C_0)$ and $V_c^{\pi}(p_x',C)$ as $U_c^{\pi}(p_x',C) -  U_c^{\pi}(p_x',C_0)$, we have
    \begin{equation}\label{eq:value_pi}
        \begin{aligned}
        V^{\pi}(C) = \ &  \gamma V^{\pi}(f(C)) + \lambda_j E_{X}\{ \\
        & \pi_{xj}(p_x,\mathbf{u}_{xj},C) [ u_{xj} + E_{X'|X}V_c^{\pi}(p_x',C) -  \gamma V^{\pi}(f(C))]  - \pi_{xj}(p_x,\mathbf{u}_{xj},C_0)  u_{xj} \}, \\
        V_c^{\pi}(p_x',C) = \ & \sum_{c\in C}\pi_c(c|p_x',C) [r(p_x',c) +  \gamma V^{\pi}(f(C,c))], \\
        V^{\pi}(C_0) = \ & 0.
        \end{aligned}
    \end{equation}
\end{proposition}

The proof of proposition \ref{thm:p1} can be found in Appendix B. The new variable $V^{\pi}(C)$ can be interpreted as the net utility gain from the coupon set. In the following discussion, we call $V^{\pi}(C)$ the ``\textit{ex ante} value function'' and $V_c^{\pi}(p_x',C)$ the ``\textit{ex post} value function''. As we already tackle the regularity problem with proposition \ref{thm:p1}, in the following exposition we always assume $\gamma = 1$.

From Equations \eqref{eq:util_ex_ante}, \eqref{eq:util_interim}, \eqref{eq:util_ex_post}, and \eqref{eq:value_pi} it is not hard to derive the form of the optimal policies $\pi_c^{*}$ and $\pi_{xj}^*$
\begin{equation}\label{eq:policy_opt}
    \begin{aligned}
        \pi_c^{*}(c|p_x',C) & = I(c = \arg\max_{c\in C}\{r(p_x',c) +  V^*(f(C,c))\}), \\
        \pi_{xj}^*(t_y,\mathbf{u}_{xj},C) = \pi_{xj}^*(t_y,u_{xj},C) & = I(u_{xj} + E_{X'|X}V_c^{*}(p_x',C) -  V^*(f(C)) \geq 0), \\
    \end{aligned}
\end{equation}

\noindent and the corresponding optimal value functions $V^*$ and $V^*_c$
\begin{equation}\label{eq:value_opt}
    \begin{aligned}
        V_c^*(p_x',C) & = \max_{c\in C}\{r(p_x',c) +  V^*(f(C,c))\}, \\
        V^*(C) & =  V^*(f(C)) + \lambda_j E_{X} [\max(u_{xj} + E_{X'|X}V_c^*(p_x',C) -  V^*(f(C)),0) - \max(u_{xj},0)], \\
        V^*(C_0) & = 0,
    \end{aligned}
\end{equation}

\noindent where $I(\cdot)$ is the indicator function: $I(X)=1$ if and only if the statement $X$ is true.

One more technical problem remains. To compute $V^*$ exactly, we need to know parameter $\lambda_j$ and distributions $P_j,P_j'$ accurately. This is not possible for observers like us; luckily, we can derive lower and upper bounds of $V^*$ as follows.

\begin{proposition}\label{thm:p2}
    Consider
  \begin{equation}\label{eq:value_approx}
      \begin{aligned}
          V^L(C) & =  V^L(f(C)) + E_{x',p_x'|C_0}I(x'=1)[V_c^L(p_x',C) -  V^L(f(C))], \\
          V_c^L(p_x',C) & = \max_{c\in C}[r(p_x',c) + V^L(f(C,c))]; \\
          V^U(C) & =  V^U(f(C)) + E_{x',p_x'|C}I(x'=1)[V_c^U(p_x',C) -  V^U(f(C))], \\
          V_c^U(p_x',C) & = \max_{c\in C}[r(p_x',c) + V^U(f(C,c))],
      \end{aligned}
  \end{equation}
  \noindent with $V^L(C_0) = V^U(C_0) = 0$ and $x'$ being the binary indicator of whether there is a trip served by the target mobility service within the time step. We have $V^L(C) \leq V^*(C) \leq V^U(C)$ for all $C \in \mathcal{C}$.
\end{proposition}

The proof of proposition \ref{thm:p2} can be found in Appendix B. One important property of these approximations is that $\lambda_j$, $P_j$, and $P_j'$ do not show up explicitly in Equation \eqref{eq:value_approx}. Therefore, we can approximate $V^L$ and $V^U$ with the estimation of the joint distribution of $x'$ and $p_x'$ from data. This leads to the following approximation of $V^*$:
\begin{equation}\label{eq:value_comp}
    \begin{aligned}
        \hat{V}(C) & = \hat{V}(f(C)) + \hat{\lambda}_jE_{p_x'|\hat{P}_j}[\hat{V}_c(p_x',C) -  \hat{V}(f(C))], \\
        \hat{V}_c(p_x',C) & = \max_{c\in C}[r(p_x',c) +  \hat{V}(f(C,c))], \\
        \hat{V}(C_0) & = 0.
    \end{aligned}
\end{equation}

\noindent where $\hat{\lambda}_j$ is the estimated service selection probability $P(x'=1)$ and $\hat{P}_j$ is the estimated marginal distribution of fare $p_x'$.

We point out that the formulation \eqref{eq:value_comp} is the most natural way to estimate the long-term value of coupons given only information on $\hat{\lambda}_j$ and $\hat{P}_j$. Therefore, the arduous deductions above mostly provide a sanity check that the straightforward modeling approach \eqref{eq:value_comp} is indeed valid under certain regular conditions.

Next, we use a simple example to illustrate how we can practically compute $\hat{V}$ with equation \eqref{eq:value_comp}. This method can be generalized to the computation of value functions in equations \eqref{eq:value_opt} and \eqref{eq:value_approx}.

\textbf{Example 1.} Consider a setting in which the service selection rate is $\hat{\lambda}_j = 0.05$ and the fare distribution $\hat{P}_j$ can be described as $\log(p_x') \sim N(3.15,0.75^2)$. Assume now the traveler has two available coupons $C = \{\langle 5,3,1 \rangle, \langle 10,2,1 \rangle \}$.

First, we enumerate all coupon sets that are possible in the future. Under the current simplified setting, the complete list can be given explicitly: $C_0$, $\{c_0,\langle 5,1,1 \rangle \}$, $\{c_0,\langle 5,2,1 \rangle \}$, $\{c_0,\langle 10,1,1 \rangle \}$, $\{c_0,\langle 5,2,1 \rangle,\langle 10,1,1 \rangle \}$.

Next, we compute $\hat{V}(C')$ for each possible $C'$. We begin with $C' = \{c_0,\langle 5,1,1 \rangle \}$; since $f(C') = C_0$, with equation \eqref{eq:value_comp} we have

\begin{equation}
    \begin{aligned}
        \hat{V}_c(p_x',C') & = \max_{c\in C'}r(p_x',c) = \min(p_x',5), \\
        \hat{V}(C') & = \hat{\lambda}_jE_{p_x'|\hat{P}_j}\hat{V}_c(p_x',C) = 0.05 E_{p_x'\sim N(3.15,0.75^2)}\min(p_x',5),
    \end{aligned}
\end{equation}

\noindent where the expectation in the second equation can then be computed numerically by sampling methods. Usually, direct sampling with 1,000 samples are enough for a good accuracy, but one can use more advanced methods such as importance sampling to improve computation efficiency. Similarly, we can compute $C' = \{c_0,\langle 10,1,1 \rangle \}$ directly.

Then, we illustrate how to compute $C'=\{c_0,\langle 5,2,1 \rangle \}$. We first notice that value $\hat{V}(f(C')) = \hat{V}(\{c_0,\langle 5,1,1 \rangle \})$ is already available according to the above computations. We now apply equation \eqref{eq:value_comp} again and have

\begin{equation}
    \begin{aligned}
        \hat{V}_c(p_x',C') & = \max(\min(p_x',5),\hat{V}(f(C'))), \\
        \hat{V}(C') & = \hat{V}(f(C')) + \hat{\lambda}_jE_{p_x'|\hat{P}_j}[\hat{V}_c(p_x',C) - \hat{V}(f(C'))] \\
        & = \hat{V}(f(C')) + 0.05 E_{p_x'\sim N(3.15,0.75^2)}\max(\min(p_x',5)-\hat{V}(f(C')),0),
    \end{aligned}
\end{equation}

\noindent which can then be evaluated by sampling methods.

Finally, one can notice that by using the same approach, we can compute the value of $V(C)$ by sampling methods when we know the value $V(f(C,c))$ for every $c \in C$. Since the transition $f$ is unidirectional in time, this method is always feasible by backward deduction.\hfill $\square$

The next example shows that when coupons have small promotional effects on the service selection rate $P(x'=1)$ (e.g., the utility gain from coupons $E_{X'|X}V_c^*(p_x',C) - V^*(f(C))$ is much smaller than the utility gain from trips $u_{xj}$), bounds $V^L$ and $V^U$ are close to each other and therefore $\hat{V}$ is a fairly accurate approximation.

\textbf{Example 2.} Consider a setting in which the default service selection rate $P(x' = 1|C_0) = 0.05$, the derivative of the service selection rate with respect to the coupon value $\frac{\partial}{\partial V}P(x' = 1) = 0.01$, and the fare distribution $\log(p_x') \sim N(3.15,0.75^2)$. Suppose now the traveler has a set of coupons $C = \{c_0,\langle 10,30,2 \rangle, \langle 10,15,1 \rangle, \langle 5,20,1 \rangle, \langle 20,5,1 \rangle \}$.

Figure \ref{fig:value_sim_multi} shows the estimated value functions: the subfigure (a) shows the traveler's value functions $V^U(f^{(T)}(C))$ and $V^L(f^{(T)}(C))$ with respect to time $T$ and the coupon set $C$; the subfigure (b) shows the differences in value functions $\Delta V = V(f^{(T)}(C)) - V(f^{(T-1)}(f(C,c)))$ with respect to the coupon group $c = \langle 10,30,2 \rangle$ for both $V^L$ and $V^U$.
From both subfigures, we can see that the gap between the upper bound and the lower bound is small. \hfill $\square$

\begin{figure}[!t]
\centering
\subfigure[$V$]{
  \includegraphics[width = 0.45\textwidth]{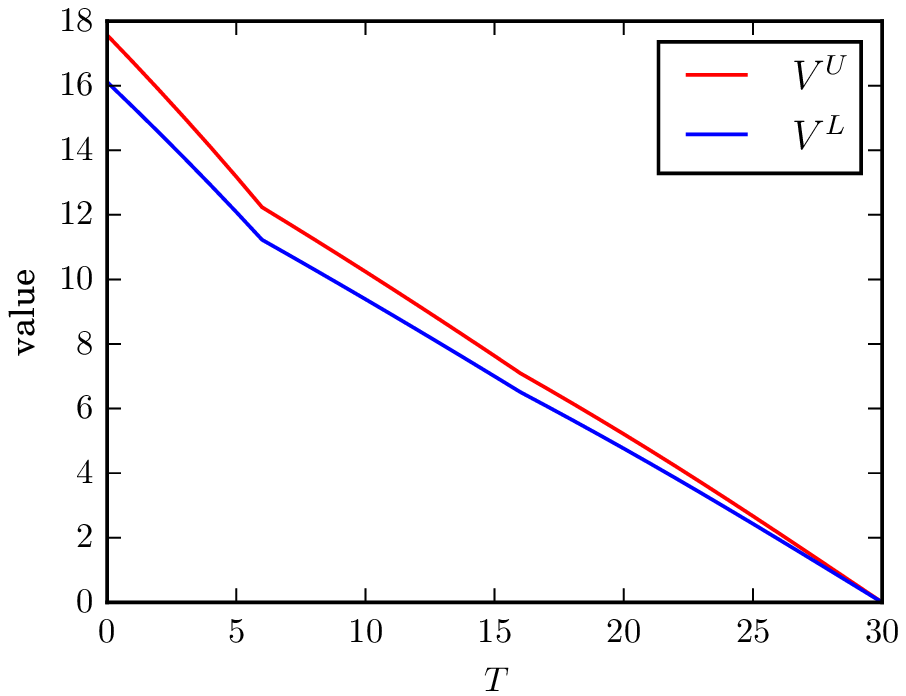}
}
\subfigure[$\Delta V$]{
  \includegraphics[width = 0.45\textwidth]{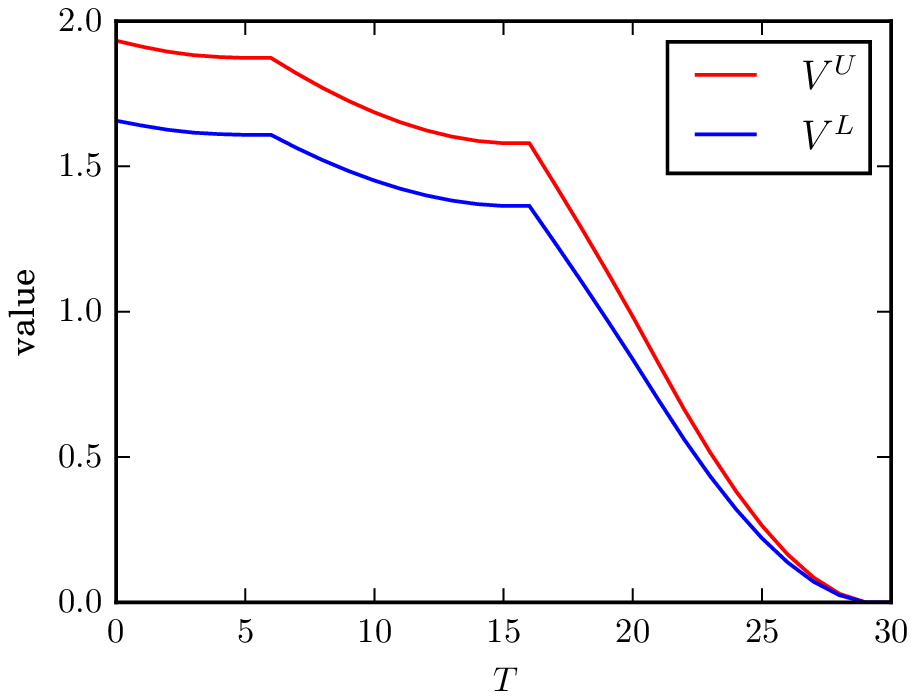}
}
\caption{Value function approximations}\label{fig:value_sim_multi}
\end{figure}

Before ending this subsection we point out another implication of equation \eqref{eq:value_approx} on the error structure in estimations of the traveler's coupon selection behavior. If we assume rationality from the traveler, the coupon choice probability $P(c|p_x',C)$, from the perspective of the observer, is
\begin{equation}\label{eq:choice_prob}
    P(c|p_x',C) = E_{V|D} I(c = \arg\max_{c\in C} \{r(p_x',c) + V(f(C,c))\} ),
\end{equation}

\noindent where $V$ is estimates of the optimal value function $V^*$ given data $D$. In specifying the error structure of $V$, people usually assume the additive separability $V = \hat{V} + \varepsilon_V$, and that $\varepsilon_V$ is normally distributed (probit model) or follows Gumbel (Type-I extreme) distribution (logit model). The major reason for such choices is to obtain tractable analytic forms of the choice probability $P(c|p_x',C)$.
However, because in our model the value $V$ has both an upper bound $V^U$ and a lower bound $V^L$, our specification of $\varepsilon_V$ should have bounded support. Specifications with probit or logit models can lead to erroneous results.

\section{Data and Observations}
\label{sec:data}

\subsection{Data description}

Our data comes from a car-sharing service in Shanghai, China and is collected during the period from January 2017 to July 2017.
The dataset includes activities from 0.16 million registered travelers and contains more than 1.5 million trip records. Figure \ref{fig:order_vs_user} shows that both the daily order volume and the number of registered travelers increased steadily during the period, but the daily order volume per traveler exhibited significant changes in its pattern on April 1st and June 10th. Such changes are due to external events including the announcement of a new pricing scheme and summer holidays. Since these external events potentially result in changes of traveler behavior, to avoid estimation bias we focus on the stable period from April 2017 to June 2017 in the discussion below. We notice that the stability of operation during the concerned period indicates that there is no significant change in the operations of alternative travel modes and traveler's preference of the target mobility service, and partly justifies our previous assumptions \ref{thm:ag2} and \ref{thm:ag3}.

\begin{figure}[!t]
\centering
\includegraphics[width = 0.5\textwidth]{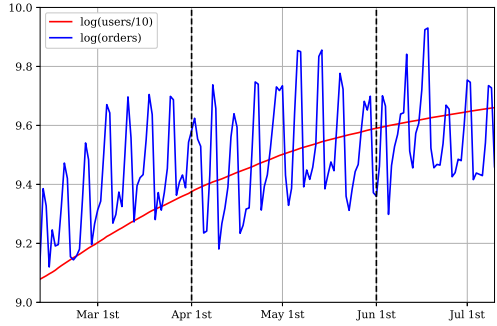}
\caption{Daily order volume and number of registered travelers}\label{fig:order_vs_user}
\end{figure}

\begin{figure}[!t]
\centering
\subfigure[Trip quantities of users]{
  \includegraphics[width = 0.4\textwidth]{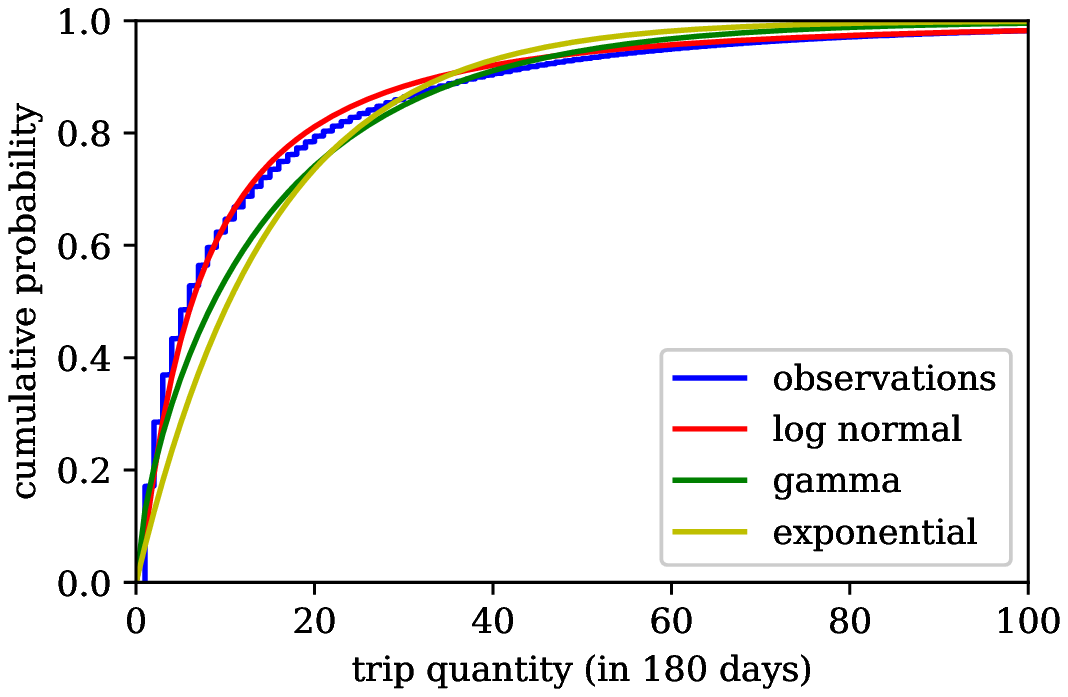}
}
\subfigure[Trip fare]{
  \includegraphics[width = 0.4\textwidth]{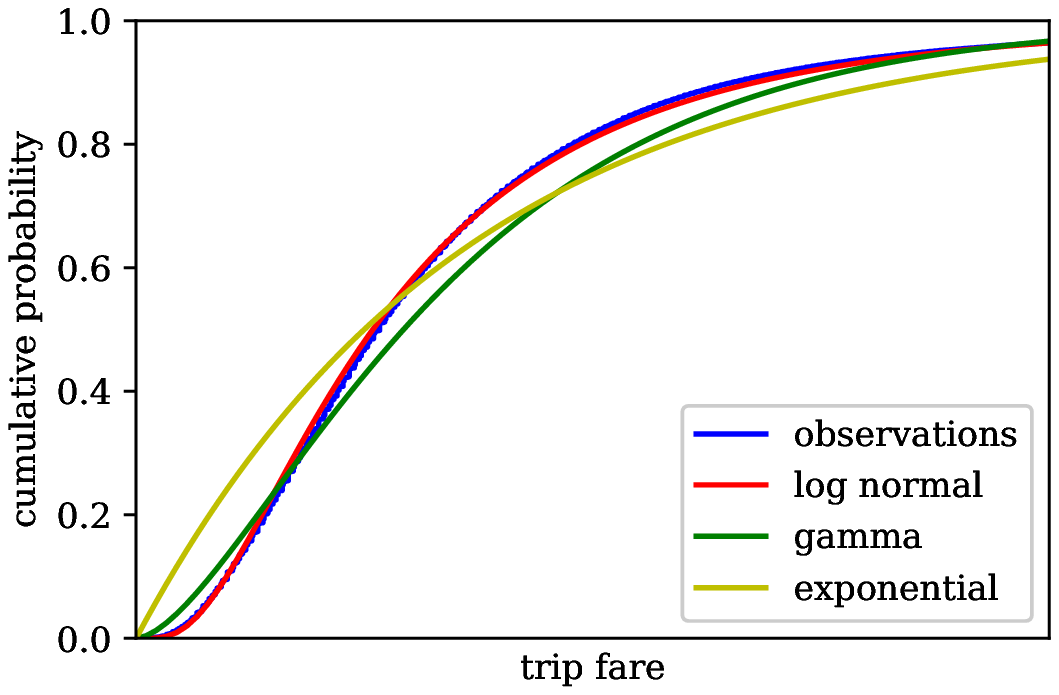}
}
\caption{Distribution of major statistics}\label{fig:dist}
\end{figure}

From part (a) of Figure \ref{fig:dist} we can see that most of the registered travelers are infrequent users of the car-sharing service: more than 80\% of them use the service less than once a week. Therefore, for most travelers, the temporal correlation of trips between consecutive days is weak, and our choice of time step size (1 day) is acceptable.

Next, we briefly introduce the coupon distribution mechanism adopted by this car-sharing service during the period from April 2017 to June 2017. Recall that we only model the impact of coupons in general forms in previous sections; however, the specific coupon distribution and redemption process can influence how travelers perceive the existence of coupons and therefore shape the distribution of the observed data.

First, when the operator wants to give a traveler $j$ a coupon $\tilde{c}$, she adds this coupon into the traveler's e-wallet at the server side. Later, when the traveler opens or refreshes her app, the app communicates with the remote server and updates the local e-wallet record. However, this update information does not pop up automatically on the main screen.

Secondly, during the trip requesting stage, the coupon information, such as how many coupons the traveler has currently, does not show up automatically, although the traveler can check her e-wallet at any time.
The traveler is able to redeem a coupon only when she finishes a trip and is in the payment stage; however, coupon redemption is not made a default option, and to redeem a coupon the traveler has to open her e-wallet from the payment panel and make an explicit selection.
All these interaction designs make the coupon redemption process less obvious for the travelers.

During the concerned period, all the coupons distributed by the service operator are of the same type: each coupon can be used freely before expiration and can reduce the trip fare by at most its face value. This is in accordance with assumption \ref{thm:agg}. Moreover, there are only four distinct face values: 5, 10, 20, and 30. The valid period of a coupon can be as short as 3 days or as long as 90 days.

Finally, the operator does not employ any sophisticated coupon distribution strategy during the concerned period. The operation team distributes coupons manually and the pattern is highly unpredictable due to a changing budget constraint. Moreover, this coupon distribution process does not differentiate between frequent and infrequent travelers. Therefore, the arrival of coupons can be viewed as a random exogenous event and we do not consider it in our model.

\subsection{Data processing}

The raw dataset consists of order and coupon records: each order record consists of details of a trip completed via the car-sharing service, and each coupon record includes a coupon's face value, start date, expiration date, and corresponding traveler. Table \ref{tab:data_des} summarizes the data fields used in our study.

\begin{table}[ht]
    \small
\centering
\begin{tabular}{|c|c|}
\hline
Type & Fields \\ \hline
\multirow{2}{*}{Order} & order ID, (encrypted) traveler ID, trip start time, trip end time, \\
 & fare, used coupon ID, payment \\ \hline
Coupon & coupon ID, (encrypted) traveler ID, coupon face value, start time, expire time \\ \hline
\end{tabular}
\caption{Data fields}\label{tab:data_des}
\end{table}

This raw dataset is processed for the subsequent coupon selection analysis. The processed dataset consists of three categories of features: trip details, including realized fare $p_x'$ and redeemed coupon $\tilde{c}$ (along with its coupon group $c$); coupon details, including available coupon set $C$ and coupon attention state $I_a$ (which will be discussed in next subsection); and traveler-specific details, including average daily trip frequency with the car-sharing service $\hat{\lambda}$, and average and standard deviation of the logarithm of the trip fares $\hat{\mu}_p,\hat{\sigma}_p$. The latter two are used to construct the empirical marginal distribution of trip fare $\hat{P}_j$: $\log (p_x') \sim N(\hat{\mu}_{pj},\hat{\sigma}_{pj}^2)$. We choose the log-normal distribution because it adequately matches the empirical distribution of trip fares, as shown in part (b) of Figure \ref{fig:dist} (the fare details are removed from this figure per request of the service operator).
The available coupon set $C$ for each trip order is recovered from the coupon records. More specifically, we track the life cycle of each coupon $\tilde{c}$ and append it to the set $C$ when it is still alive.

Because our major concern in this study is the coupon selection behavior, we remove records in which the coupon set is empty ($C=C_0$). The final dataset on coupon selection contains more than 0.36 million records and approximately $0.08$ million of them have only one coupon in the set $C$. As a reference, the dataset constructed from all available data contains more than 0.9 million records and 0.23 million of them have only one coupon in $C$.

\subsection{Observations on the coupon redemption behavior}

In this subsection, we discuss several observations from the coupon selection data and introduce a few hypotheses for explanations.

First, consider the case in which the traveler has only one coupon. On one hand, from equation \eqref{eq:value_opt} we know that the optimal value $V^*$ is upper bounded by $v$. So if $p_x' \geq v$ and there are temporal discounting effects, the expected value from future redemption is strictly smaller than the value of immediate redemption, $v$, and a utility-maximizing traveler should always redeem the coupon.
On the other hand, the observed coupon redemption ratio when $p_x' \geq v$ is consistently lower than 1, as shown in Figure \ref{fig:v_over_p_rate_single}.
In Figure \ref{fig:v_over_p_rate_single}, we group data with different $p_x'/v$ with a granularity of 0.2 and each point on the curve represents the statistics of the data from the group on its left: for example, point at $p_x'/v = 1$ and $v = 5$ represents statistics for all data with $0.8 \leq p_x'/v < 1, v = 5$, and point at $p_x'/v = 1.2$ and $v = 10$ summarize statistics for all data with $1 \leq p_x'/v < 1.2, v = 10$, etc.

\begin{figure}[H]
\centering
  \includegraphics[width = 0.4\textwidth]{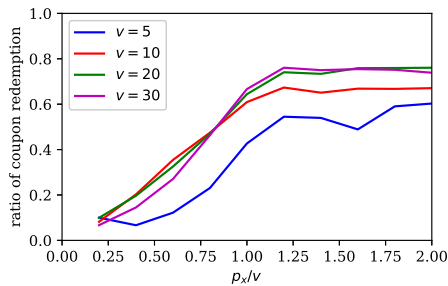}
\caption{Ratio of coupon redemption v.s. fare-value ratio, cases with only one coupon}\label{fig:v_over_p_rate_single}
\end{figure}

Next, we consider the case in which there is a coupon $\langle v,T \rangle$ in the available coupon set $C$ such that the trip fare $p_x'$ satisfies $p_x' \geq v$.
By the same argument as above, we expect to observe that the coupon redemption ratio equals to 1 in this case.
However, this is not true for the observed data, as shown in Figure \ref{fig:v_over_p_rate}.
Figure \ref{fig:v_over_p_rate} shows the relationship between the coupon quantity and the coupon redemption ratio of the observed data.
From part (a), we can see that the coupon redemption ratio is consistently below 1. Nevertheless, this gap diminishes as the coupon quantity increases. From part (b-c) we further observe that traveler specific factors such as past experiences and trip frequencies also have strong impacts on the coupon redemption ratio.

\begin{figure}[ht]
\centering
\subfigure[Overall]{
  \includegraphics[width = 0.4\textwidth]{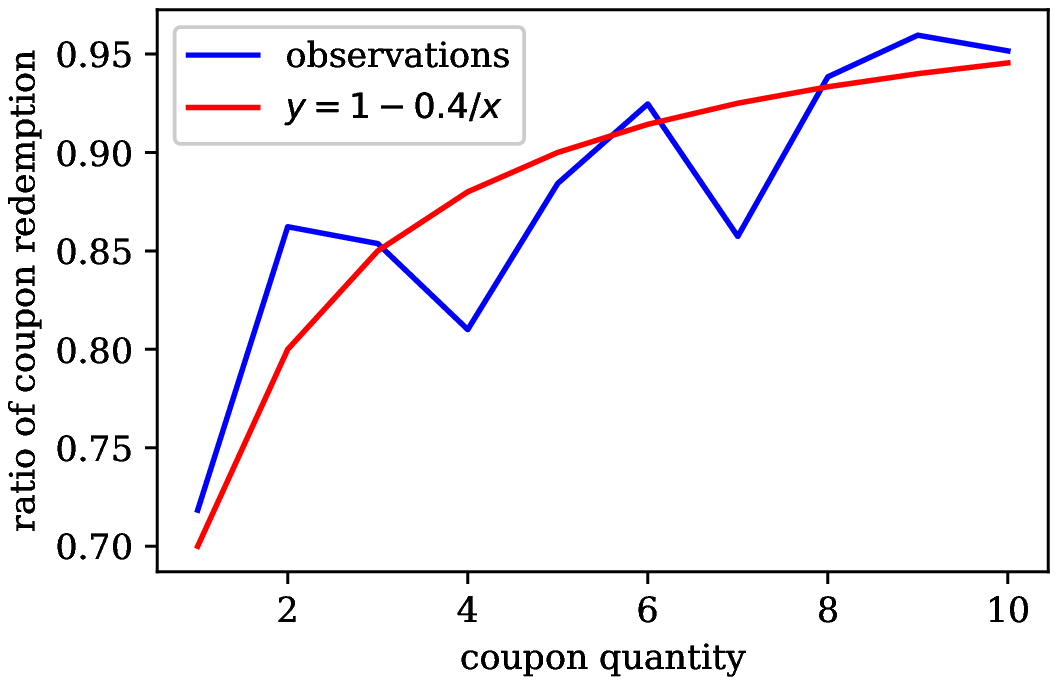}
}
\subfigure[By experience]{
  \includegraphics[width = 0.4\textwidth]{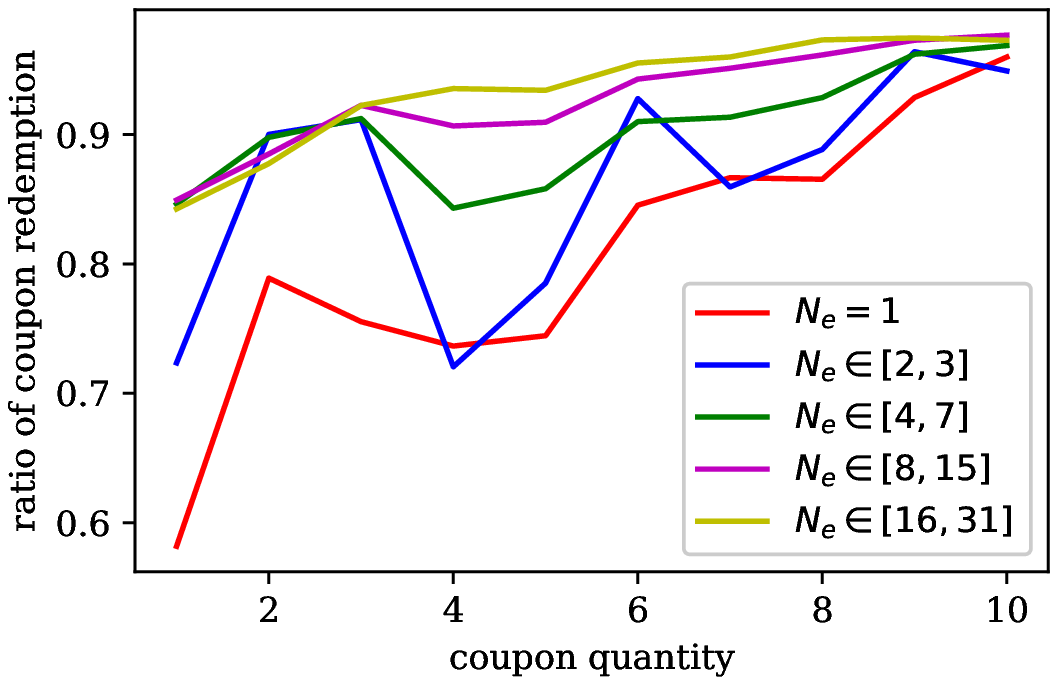}
}
\subfigure[By frequency]{
  \includegraphics[width = 0.4\textwidth]{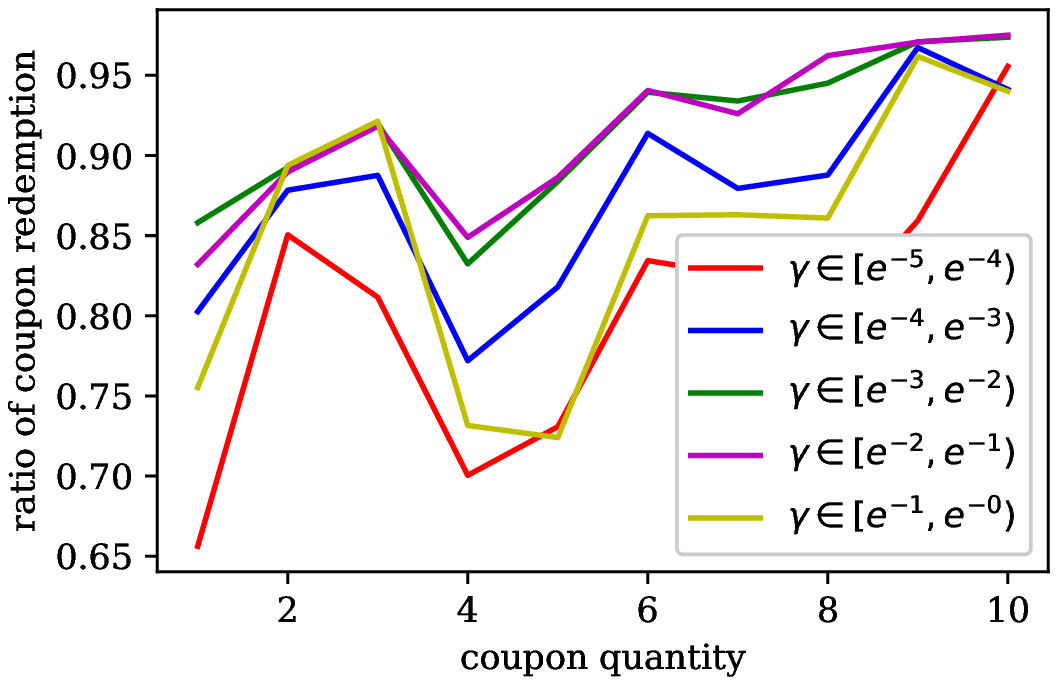}
}
\caption{Ratio of coupon redemption v.s. coupon quantity, cases in which some coupon face values $v$ exceed fare $p_x'$}\label{fig:v_over_p_rate}
\end{figure}

Now, we introduce several hypotheses to explain the above observations. We limit our attention to approaches in two directions: either to relax the optimality assumption or to examine and redefine the decision problem that the traveler optimizes.

We point out that the deviation from optimality in the cases with only one coupon greatly limits our choices in the first direction. In situations that the traveler has only one coupon, the traveler is confronted with much easier decision problems compared with general cases. Models on the computation complexity of the decision problem, such as those on bounded rationalities or deliberate attention, then imply that decisions in the cases with only one coupon are more likely to follow the optimal decision $\pi^*$. However, above observations show that the deviation is the greatest in this setting.

One possible option in this direction is to consider near-optimal stochastic policies, such as the entropy-regularized optimal policy $\pi_H$ \cite{ziebart2008maximum}

\begin{equation}\label{eq:stoc_policy}
    \pi_{H,c}(c|p_x',C) \propto \exp[r(p_x',c) + V^{\pi_H}(f(C,c))].
\end{equation}

Clearly stochastic policies such as $\pi_H$ can explain a coupon redemption ratio below 1 when we have $p_x' \geq v$. However, coupon redemption ratios from such policies are sensitive to the coupon value $v$, at least in the cases with only one coupon. We can illustrate the intuition of this claim by looking at equation \eqref{eq:value_comp_single}: given the same $\hat{\lambda_j},\hat{P_j}$ and $T$, one can show by induction on $T$ that the gap between $v$ and $\hat{V}(v,T)$ increases monotonically as $v$ increases.
Therefore, the redemption ratio under the condition $p_x' / v \geq 1$ should also increase monotonically with $v$. However, this is not fully consistent with our observations in Figure \ref{fig:v_over_p_rate_single}: the redemption ratio under $v = 30$ is not significantly greater than the one under $v = 20$.

We can follow the discussion above and keep searching for the ``correct'' policy; however, in this process we need to make more assumptions and the generalization power of our results diminishes.

Next, we consider approaches in the second direction: the decision problem faced by the traveler is different from the one in our model. For our specific problem of coupon selections, we further assume that such difference comes from the knowledge of the available coupon set $C$. As discussed in previous sections, the operator does not provide much information about the arrival of coupons, so there are possibilities that the traveler has only partial knowledge of $C$.

To identify variables that capture the traveler's knowledge of $C$, we recall that a traveler can be aware of a coupon $\tilde{c}$ only when she opens her e-wallet during the coupon's life cycle. Since one motivation of such action is to use a coupon for payment, we consider the following variable: ``activation record'' $I_a(\tilde{c})\in \{0,1\}$ of coupon $\tilde{c}$, which equals to 1 if and only if there is a past realized trip, after which the traveler selected a coupon $\tilde{c}' \neq \tilde{c}$ for payment and coupon $\tilde{c}$ was in the e-wallet at that time.

We now examine explanation power of $I_a$ on the aforementioned deviations. Figure \ref{fig:v_over_p_f_rate_single} shows the relationship between $p_x'/v$ and the coupon redemption ratio in cases that the traveler has only one coupon $\tilde{c}$ and the coupon satisfies $I_a(\tilde{c}) = 1$. Figure \ref{fig:v_over_p_f_rate} shows the relationship between the coupon quantity and the coupon redemption ratio in cases that there is a coupon $\tilde{c} = \langle v,T \rangle$ in the available coupon set $C$ such that the trip fare $p_x'$ satisfies $p_x' \geq v$ and $I_a(\tilde{c}) = 1$.
From both figures, we can see that the coupon redemption ratio becomes much closer to 1 when we impose the restriction on $I_a$.
Moreover, part (b-c) of Figure \ref{fig:v_over_p_f_rate} indicates that the impact of past experience and trip frequency diminishes under the same restriction.

\begin{figure}[H]
\centering
  \includegraphics[width = 0.4\textwidth]{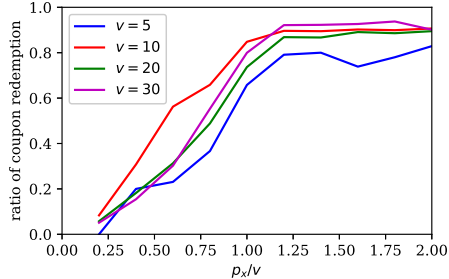}
\caption{Ratio of coupon redemption v.s. fare-value ratio, cases with only one coupon, $I_a(\tilde{c}) = 1$}\label{fig:v_over_p_f_rate_single}
\end{figure}

\begin{figure}[H]
\centering
\subfigure[Overall]{
  \includegraphics[width = 0.4\textwidth]{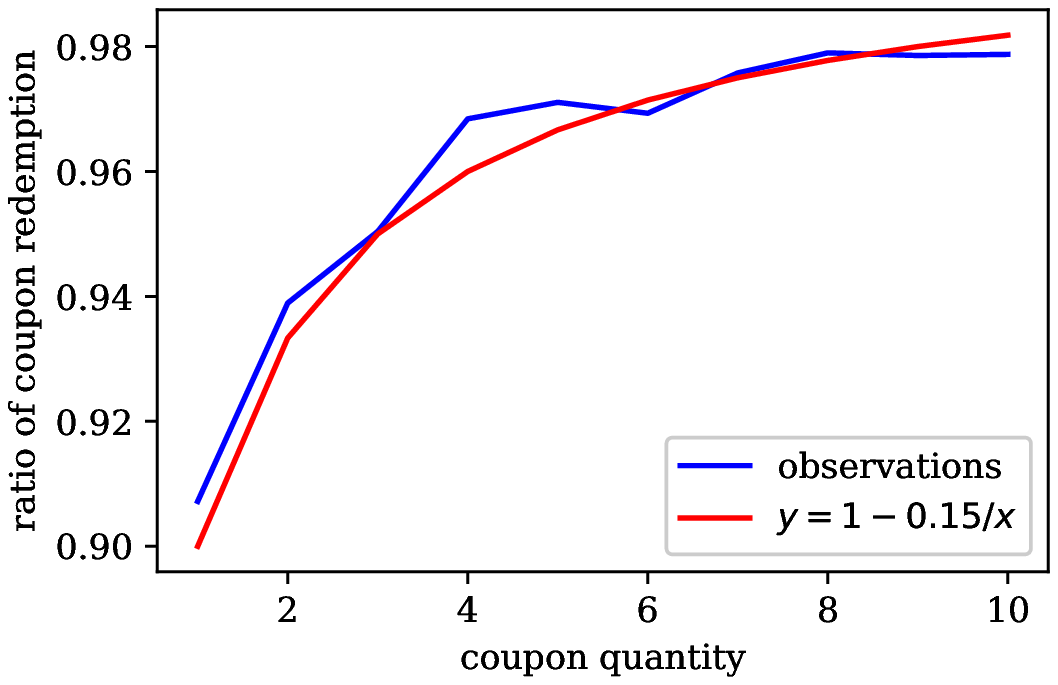}
}
\subfigure[By experience]{
  \includegraphics[width = 0.4\textwidth]{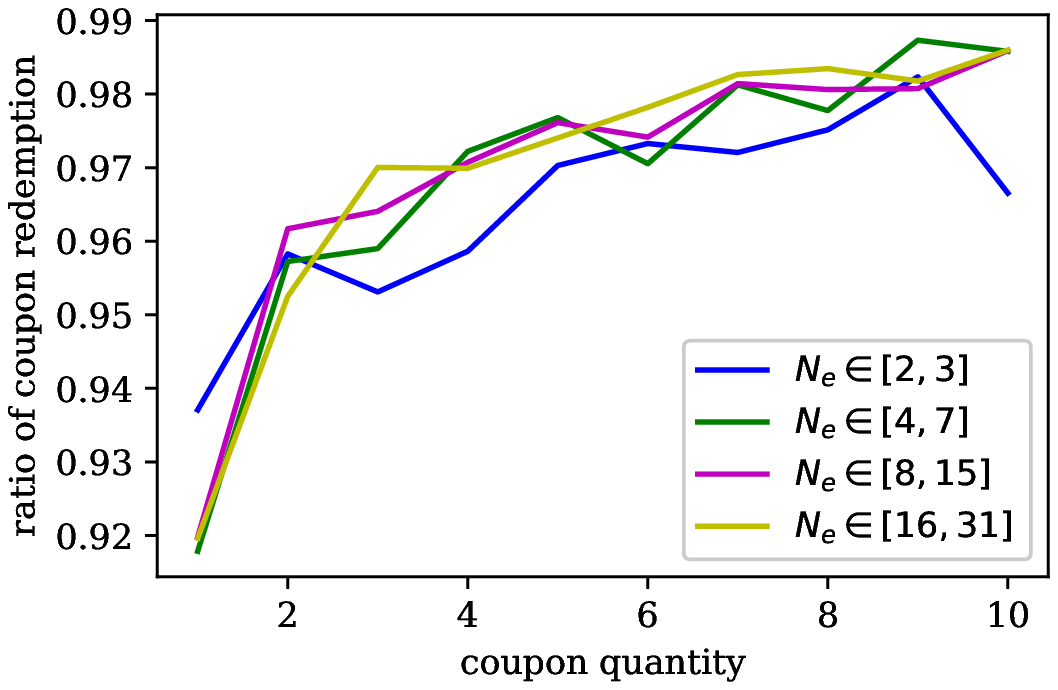}
}
\subfigure[By frequency]{
  \includegraphics[width = 0.4\textwidth]{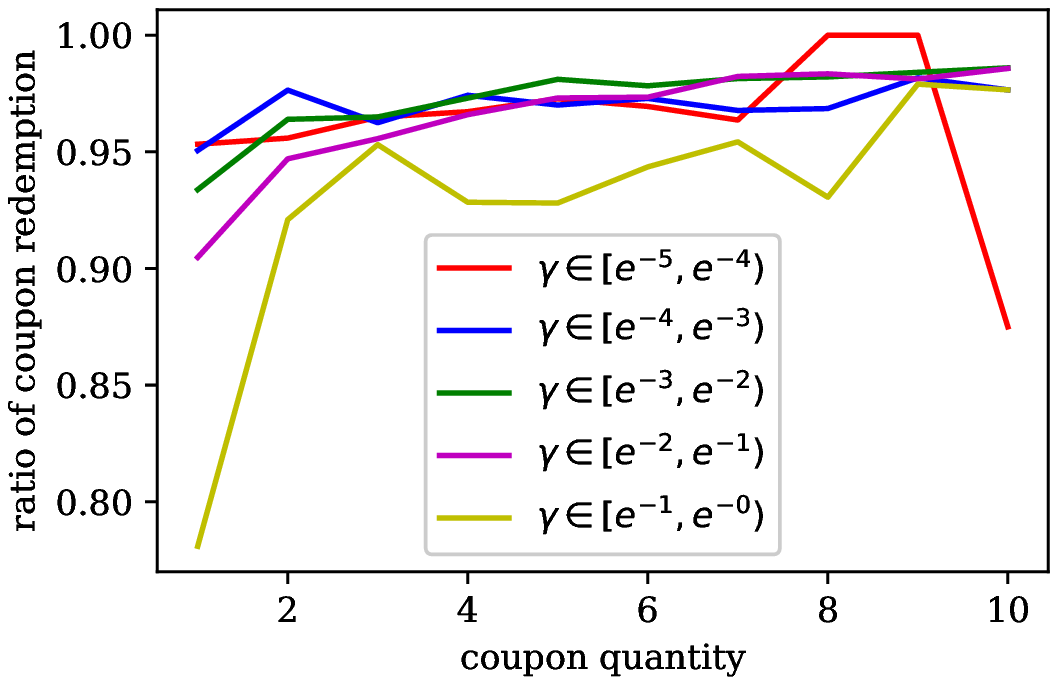}
}
\caption{Ratio of coupon redemption, cases in which some coupon face values $v$ exceed fare $p_x'$ and $I_a(\tilde{c}) = 1$}\label{fig:v_over_p_f_rate}
\end{figure}

\section{Model with Unawareness}
\label{sec:att_model}

In this section, we formalize the notion of unawareness. The basic idea is that travelers, influenced by external circumstances and internal limited memories, may make decisions with respect to a perceived (awareness) subset $C_a$ rather than the whole coupon set $C$.

\subsection{Model formulation}

If a traveler is only aware of the subset $C_a$ and make decisions with respect to it, it is unlikely that the traveler concerns about the unawareness itself. This suggests the following assumption:

\begin{assumption}\label{thm:att}
    The traveler does not consider her unawareness of the set of available coupons in evaluating future coupon redemption values. Therefore, This evaluation is the same as in the baseline model.
\end{assumption}

With this assumption, in following discussions we can consider the same policies $\pi_{xj}$ and $\pi_c$ as those developed in Section \ref{sec:model}.

Next, we discuss the impact of unawareness on decisions under the setting in which traveler $j$ has an available coupon set $C$. Because unawarenesses in different periods are correlated, we introduce a new state variable $S_a$ for it. After each trip with the target mobility service, this state $S_a$ is updated with the available coupon set $C$ and the selected coupon group $c$, according to a function $f_a$:
\begin{equation}
    S'_{a} = f_a(S_a,C,c).
\end{equation}

At the beginning of a trip, the awareness set $C_a$ is generated according to the state of attention $S_a$ and a distribution $P_a$: $P(C_a) = P_a(C_a|C,S_a)$. Given $X$ and $C_a$, the traveler selects a travel mode according to $\pi_{xj}$:
\begin{equation}
    P(i=1) = \pi_{xj}(p_x,\mathbf{u}_{xj},C_a).
\end{equation}

If the traveler completes the trip via the mobility service, detail of the realized trip $X'$ is revealed according to $X' \sim P'_{j}(\cdot|X)$ and the traveler selects a coupon for redemption given the realized trip fare $p_x'$.
However, with unawareness, the coupon selection behavior is complicated and can be described in two steps. First, the traveler makes a choice $c$ with respect to the awareness set $C_a$ according to the probability $\pi_c(c|p_x',C_a)$. Then, if there is no coupon redemption ($c = c_0$), the traveler proceeds to payment; otherwise ($c \neq c_0$), the traveler must open her e-wallet. In this case, she will find that she has the whole coupon set $C$; we call this situation \textit{attention recovery}. She then re-evaluates and makes her final decision with respect to set $C/C_0$.

Here, the default option $c_0$ is removed from the final consideration set for regularity. In fact, if the traveler is aware of all coupons $C_a=C$, the re-evaluation should have no effect on the coupon selection probabilities. However, if $c_0$ is included in the final consideration set, the probability of having no coupon redemption increases and differs from $\pi_c(c_0|p_x',C)$.

In summary, the probability of selecting coupon group $c$ is
\begin{equation}\label{eq:coupon_prob_att}
    P(c|p_x',C_a) =
    \begin{cases}
        \pi_c(c_0|p_x',C_a) & c = c_0 \\
        \pi_c(c|p_x',C/C_0) (1 - \pi_c(c_0|p_x',C_a)) & \text{otherwise}
    \end{cases}
\end{equation}

Now, if the traveler follows the utility-maximizing policy $\pi^*$, the coupon selection probability from the perspective of the observer can be expressed as
\begin{equation}\label{eq:choice_prob_attention}
    \begin{aligned}
    P(c|p_x',C,S_a) = \ & \sum_{C_a\in \mathcal{A}(C)}P_a(C_a|C,S_a)P(c|p_x',C_a) \\
    P(c_0|p_x',C_a) = \ & E_{V|D} I(c_0 = \arg\max_{c\in C_a} \{r(p_x',c) + V(f(C_a,c))\} ) \\
    P(c|p_x',C_a) = \ & E_{V|D} I(c_0 \neq \arg\max_{c\in C_a} \{r(p_x',c) + V(f(C_a,c))\} ) \\
    & \cdot I(c = \arg\max_{c\in C/C_0} \{r(p_x',c) + V(f(C,c))\} ). \\
    \end{aligned}
\end{equation}

\subsection{Specification of the inattention mechanism}
\label{subsec:att_model_spec}

In this subsection, we specify the form of the state of attention $S_a$ and the awareness set probability function $P_a$.

First, from the state transition of $S_a$ we see that $S_a$ can be interpreted as a function of coupon set $C$. However, this function space is too large to be considered practically: we need to consider every possible correlation among coupon groups. Here, we restrict our attention to the coupon-group-level features. That is, we can interpret $S_a$ as a function from coupon groups $c$ to features $S_a(c)$ related only to that coupon group.

Next, we specify the awareness set probability $P_a(C_a|C,S_a)$. We first look at the special case that there is only one available coupon $C = \{c, c_0\}$, with $c = \langle v,T,1 \rangle$. In this case, the only difference between the two possible outcomes, $C$ and $C_0$, is whether the traveler notices the coupon $\tilde{c} = \langle v,T \rangle$. Thus we can say
\begin{equation}
    P_a(C|C,S_a) = \sigma(h(S_{a}(c))), \ P_a(C_0|C,S_a) = 1 - \sigma(h(S_{a}(c))),
\end{equation}

\noindent where $\sigma$ is the sigmoid function $\sigma(x) = 1 / (1 + e^{-x})$ and $h$ is a function. A natural choice of $h$ is the linear functions $h(x) = \theta^T x + b$, where $\theta,b$ are parameters.

For the form of $P_a(C_a|C,S_a)$ in general, there are much more choices: in Subsection 2.2 we have mentioned several of them. Here, we adopt the Manski model \cite{manski1977structure} which assumes that awarenesses of different elements in the set are independent of each other. Nevertheless, the word ``element'' is still ambiguous: we can consider independence at either the coupon-level or the coupon-group-level. We now make a detailed discussion on this choice. For coupon-group-level independence, the awareness set $C_a$ can only be a subset of $C$
\begin{equation}\label{eq:att_form_group}
        P_a(C_a|C,S_a) = \prod_{c \in C_a} \sigma(h(S_{a}(c)))\cdot \prod_{c \in C/C_a} (1 - \sigma(h(S_{a}(c))));
\end{equation}

\noindent while for the coupon-level independence, the coupon groups in $C_a$ can be different from those in $C$ and the quantity in each coupon group is important. To further explain this, let us abuse the notation of set and consider $C = \{c_1,\cdots,c_m\}$ and $C_a = \{c^a_{i}|i\in I^a\}$, with $c_i = \langle v_i,T_i,n_i \rangle$ and $c^a_i = \langle v_i,T_i,n_i^a \rangle, 0 \leq n_i^a \leq n_i$. (We say this is an abuse of notation since if $n_i^a = 0$, $c^a_i$ is not an element of $C_a$.) Now the coupon-level independence leads to
\begin{equation}\label{eq:att_form}
        P_a(C_a|C,S_a) = \prod_{i = 1}^m \binom{n_i}{n_i^a} [\sigma(h(S_{a}(c_i)))]^{n_i^a} [(1 - \sigma(h(S_{a}(c_i))))]^{n_i - n_i^a}.
\end{equation}

At first sight, the coupon-group-level-independence formulation seems to be more straightforward. However, for it to make any sense, the quantity of coupon group $n$ must become one of the features in $S_a$ and has direct impact on the awareness level $h(S_{a}(c))$; otherwise, the awareness level remains the same even when we increase $n$ towards infinity. But now the specification of the relationship between $n$ and $h(S_{a}(c))$ becomes another problem. To avoid this problem, we choose the coupon-level-independence formulation instead, which provide a simple characterization of the impacts of $n_i$.

At the end of this subsection, we discuss which feature to be included in $S_{a}(c)$. A natural selection is the variable $I_a(\tilde{c})$ which represents whether the traveler has previously seen coupon $\tilde{c}$. This function can be defined at the coupon-group level because coupons in the same group are distributed at the same time and should be seen at once.
However, from Figure \ref{fig:v_over_p_f_rate_single} and \ref{fig:v_over_p_f_rate} it is shown that the feature $I_a(\tilde{c})$ cannot fully capture the inattention effect. Therefore, we add a new parameter $\theta_{a}$ to capture the remaining possibilities, such as that the traveler may forget the existence of coupons as time passes:

\begin{equation}\label{eq:att_form_single}
    \begin{aligned}
        S_{a}(c) & = I_a(c), \\
        h(S_{a}(c)) & = \theta_{a} + \theta_{as}I_{a}(c),
    \end{aligned}
\end{equation}

\noindent where $\theta_a,\theta_{as}$ are parameters.

Finally, for completeness, we provide the formulation of the state transition function $f_a$
\begin{equation}
    f_a(S_a,C,c)(f_c(c')) =
    \begin{cases}
        0 & c = c_0 \ \& \ f_c(c') \neq c_0 \ \& \ S_a(c') = 0 \\
        1 & \text{otherwise}
    \end{cases}
    , \forall c' \in C/\{c\} \cup \{\langle v,T,n-1 \rangle\}.
\end{equation}

\section{Estimation Results}
\label{sec:results}

In this section, we focus on the model estimations given travelers' coupon selection behavior dataset $\{(c_1,p_{x1}',$ $C_1,S_{a1},$ $\hat{\lambda}_1,\hat{\mu}_{p1},\hat{\sigma}_{p1}), \cdots, (c_N,p_{xN}',C_N,S_{aN},\hat{\lambda}_N,\hat{\mu}_{pN},\hat{\sigma}_{pN})\}$.
As mentioned in Section \ref{sec:data}, we focus on the data obtained during the period from April 2017 to June 2017. Estimations on the whole dataset are provided in Appendix D for reference.

Recall from Equation \eqref{eq:choice_prob_attention} that the coupon selection probability $P(c|p_x',C,S_a)$ is a mixture
\begin{equation}
    P(c|p_x',C,S_a) = \sum_{C_a\in \mathcal{A}(C)}P_a(C_a|C,S_a)P(c|p_x',C_a),
\end{equation}

\noindent where the awareness set probability $P_a(C_a|C,S_a)$ is specified in Equations \eqref{eq:att_form} and \eqref{eq:att_form_single}. If we are further given a parametric form of $P(c|p_x',C_a)$, we can estimate parameters $\theta$ in both models by maximizing the (average) log-likelihood of the model on the dataset:
\begin{equation}
        \text{Log-Likelihood} = \frac{1}{N} \sum_{l=1}^N \log P(c_l|p_{xl}',C_l,S_{al}).
\end{equation}

Concerning about the size of our dataset and the complexity of our models, we use TensorFlow to construct the computation graph of these models and the Adam algorithm \cite{Kingma_2014} to optimize the log-likelihood. Other hyper-parameters for training are summarized in Table \ref{tab:parameters}. In the estimation, the maximal training time for a model is less than an hour.

\begin{table}[h]
    \small
  \begin{center}
  \begin{tabular}{ | c | c | }
    \hline
    \textbf{Parameter} & \textbf{Value} \\ \hline
    learning rate & 0.001 \\ \hline
	mini-batch size &  256 \\ \hline
	training epochs & 50 \\ \hline
  \end{tabular}
  \caption {Hyper-parameters for training} \label{tab:parameters}
  \end{center}
\end{table}

Next, we discuss estimation results of various forms of $P(c|p_x',C_a)$. We start with estimations in cases with only one coupon.

\subsection{Cases with only one coupon}

When there is only one coupon, the available coupon set $C$ can be expressed as $\{\langle v,T,1 \rangle,c_0\}$, and we can simplify the notation of coupon $c=\langle v,T,1 \rangle$ with $\langle v,T \rangle$ and that of value function $V(\{\langle v,T,1 \rangle,c_0\})$ with $V(v,T)$. The value function $\hat{V}$ in equation \eqref{eq:value_comp} can now be computed in a simpler way
\begin{equation}\label{eq:value_comp_single}
    \begin{aligned}
        \hat{V}(v,T) & = \hat{V}(v,T-1) + \hat{\lambda}_jE_{p_x'|\hat{P}_j} \max\{\min(p_x',v) -  \hat{V}(v,T-1),0\} \\
        \hat{V}(v,T) & = 0, \ \forall \ T < 0 \ or \ v = 0. \\
    \end{aligned}
\end{equation}

Since in the cases with only one coupon the awareness set $C_a$ can only be either $C$ or $C_0$, we remain to specify the form of $P(c|p_x',C)$. With some calculations we can show that $P(c|p_x',C) = E_{V|D} I(\min(p_x',v) \geq V(v,T-1))$; since this probability only depends on $v,T$ and $p_x'$, we can denote it as $P(v,T,p_x')$ for simplicity.

In specifying $P(v,T,p_x')$, we use approximation $\hat{V}$ in Equation \eqref{eq:value_comp_single} as a feature for the optimal value function $V$. In particular, we consider an estimation with its error following the logistic distribution $V(v,T) = \theta_V \hat{V}(v,T) + \varepsilon_V, \varepsilon_V \sim Logistic(0,1/\theta_{\varepsilon})$. With this specification, we have
\begin{equation}\label{eq:model_opt_single}
    P(v,T,p_x') = \sigma(\theta_{\varepsilon} [\min(p_x',v) - \theta_V \hat{V}(v,T-1) ] ),
\end{equation}

\noindent where $\sigma$ is the sigmoid function. One can notice that this form of selection probabilities resembles the one from the ``near-optimal stochastic policies'' in equation \eqref{eq:stoc_policy}. This specification is called the ``basic specification'' in the following discussion.

The basic specification can be extended in several directions. First, in the basic specification we assume that the error $\varepsilon_V$ is independent from the face value $v$. Therefore, we have the same estimation variance in the value of a coupon with $v=30$ as in the value of a coupon with $v=5$. We can introduce the correlation between $\varepsilon_V$ and $v$ by scaling $\varepsilon_V$ with $v$: $V(v,T) = \theta_V \hat{V}(v,T) + v\varepsilon_V, \varepsilon_V \sim Logistic(0,1/\theta_{\varepsilon})$. This extension is called ``scaled'' and leads to the following coupon selection probability
\begin{equation}\label{eq:model_opt_scale_single}
    P(v,T,p_x') = \sigma(\theta_{\varepsilon} [\min(p_x',v)/v - \theta_V \hat{V}(v,T-1)/v ] ).
\end{equation}

Second, when the traveler exhibits bounded rationalities, $\hat{V}$ may not provide an accurate estimation. Therefore, we include the face value $v$ as an extra feature in our estimation: $V(v,T) = \theta_V \hat{V}(v,T) + \theta_v v + \varepsilon_V, \varepsilon_V \sim Logistic(0,1/\theta_{\varepsilon})$. This extension is called ``extra'' and leads to the following coupon selection probability
\begin{equation}\label{eq:model_opt_fix_single}
    P(v,T,p_x') = \sigma(\theta_{\varepsilon} [\min(p_x',v) - \theta_V \hat{V}(v,T-1) - \theta_v v ] ).
\end{equation}

Finally, as discussed in Section \ref{sec:model}, the true optimal value $V^*$ is always bounded by finite values. Therefore, we consider to regularize the estimate by clipping: $V(v,T) = \max\{0,\min\{v, \theta_V \hat{V}(v,T) + \varepsilon_V\}\}, \varepsilon_V \sim Logistic(0,1/\theta_{\varepsilon})$. This extension is called ``clip'' and leads to the following coupon selection probability
\begin{equation}\label{eq:model_opt_clip_single}
    P(v,T,p_x') =
    \begin{cases}
        \sigma(\theta_{\varepsilon} [\min(p_x',v) - \theta_V \hat{V}(v,T-1) ] ) & p_x' < v,\\
        1 & \text{otherwise}.
    \end{cases}
\end{equation}

These extensions are orthogonal and can be combined with one another. In total, we can obtain eight different models. Table \ref{tab:result_single_uniform} summarizes the estimated parameters and the performance of these models. In this table, the numbers are reported up to three digits after the decimal point, and 0/1 are used to denote False/True and indicate whether a condition is met. ``LL'' refers to the (average) log-likelihood. ``Accuracy'' refers to the forecasting accuracy and is computed as
\begin{equation}
    \text{Accuracy} = \frac{1}{N}\sum_{l=1}^N I(c_l = \arg\max_{c\in C_l}P(c|p_{xl}',C_l,S_{al})),
\end{equation}

\noindent ``MS'' refers to the predicted aggregate coupon redemption ratio (the ``market share'' of coupon redemption) and is computed as
\begin{equation}
    \text{MS} = 1 - \frac{1}{N}\sum_{l=1}^N P(c_0|p_{xl}',C_l,S_{al});
\end{equation}

\noindent for Table \ref{tab:result_single_uniform}, the observed aggregate coupon redemption ratio is 0.719.

\begin{table}[ht]
    \small
\centering
\begin{tabular}{|ccc|cccccc|ccc|}
\hline
\multicolumn{3}{|c|}{Inattention} & \multicolumn{6}{c|}{Utility Model and Extensions} & \multicolumn{3}{c|}{Evaluation} \\ \hline
Unaware? & $\theta_a$ & $\theta_{as}$ &Clip? & Extra? & Scaled? & $\theta_{\varepsilon}$ & $\theta_V$ & $\theta_v$ & LL & Accuracy & MS \\ \hline
0 & N/A & N/A & 0 & 0 & 0 & 0.088 & 0.620 & N/A & -0.543 & 0.758 & 0.719 \\
0 & N/A & N/A & 0 & 1 & 0 & 0.177 & 0.128 & 0.532 & -0.513 & 0.781 & 0.701 \\
0 & N/A & N/A & 0 & 0 & 1 & 1.926 & 0.513 & N/A & -0.535 & 0.751 & 0.746 \\
0 & N/A & N/A & 0 & 1 & 1 & 3.455 & 0.154 & 0.480 & -0.508  & 0.782 & 0.721  \\ \hline
1 & 1.025 & 1.850 & 0 & 0 & 0 & 0.369 & 0.743 & N/A & -0.491 & 0.780 & 0.714 \\
1 & 1.183 & 2.116 & 0 & 1 & 0 & 0.358 & 0.439 & 0.274 & -0.479 & 0.788 & 0.707 \\
1 & 1.205 & 3.868 & 0 & 0 & 1 & 4.272 & 0.652 & N/A & -0.491 & 0.777 & 0.730 \\
1 & 1.368 & 4.741 & 0 & 1 & 1 & 5.060 & 0.387 & 0.281 & -0.478 & 0.787 & 0.720 \\ \hline
1 & 1.047 & 1.450 & 1 & 0 & 0 & 0.224 & 0.826 & N/A & -0.479 & 0.785 & 0.719 \\
1 & 1.096 & 1.427 & 1 & 1 & 0 & 0.283 & 0.536 & 0.228 & -0.474 & 0.789 & 0.713 \\
1 & 1.080 & 1.449 & 1 & 0 & 1 & 4.113 & 0.797 & N/A & -0.480 & 0.787 & 0.723 \\
1 & 1.122 & 1.436 & 1 & 1 & 1 & 4.860 & 0.568 & 0.188 & -0.476 & 0.790 & 0.718 \\ \hline
\end{tabular}
\caption{Estimated parameters and performance in the case with only one coupon}\label{tab:result_single_uniform}
\end{table}

\begin{table}[ht]
    \small
\centering
\begin{tabular}{|c|cccc|}
\hline
Coupon face value & 5 & 10 & 20 & 30 \\ \hline
Occurrence & 2032 & 3319 & 50999 & 21425 \\ \hline
\end{tabular}
\caption{Occurrence of records with specific face value of the coupon, in the case with only one coupon}\label{tab:coupon_value_freq}
\end{table}

Direct estimation of models on the raw dataset can lead to bias, because our dataset is unbalanced with respect to the coupon face value $v$, as shown in Table \ref{tab:coupon_value_freq}. As $v$ is an exogenous variable and the model should not be biased toward any of its specific values, we use weights to rebalance the importance of each data record. In particular, we give weight $w=N/N(v)$ to records with coupon face value $v$, where $N(v)$ is the occurrence of records with their coupon face values equaling to $v$. The weighted log-likelihood, accuracy, and aggregate coupon redemption ratio are respectively given as
\begin{equation}
    \begin{aligned}
        \text{Weighted Log-Likelihood} & = \sum_{l=1}^N \frac{1}{N(v_l)}\log P(c_l|p_{xl}',C_l,S_{al}), \\
        \text{Weighted Accuracy} & = \sum_{l=1}^N \frac{1}{N(v_l)}I(c_l = \arg\max_{c\in C_l}P(c|p_{xl}',C_l,S_{al})), \\
        \text{Weighted MS} & = \sum_{l=1}^N \frac{1}{N(v_l)}(1 - P(c_0|p_{xl}',C_l,S_{al})), \\
    \end{aligned}
\end{equation}

\noindent where $v_l$ is the face value of the only coupon in $C_l$. Table \ref{tab:result_single_rebalanced} summarizes the estimation results with weighting. For this table, the weighted observed aggregate coupon redemption ratio is 0.715.

\begin{table}[ht]
    \small
\centering
\begin{tabular}{|ccc|cccccc|ccc|}
\hline
\multicolumn{3}{|c|}{Inattention} & \multicolumn{6}{c|}{Utility Model and Extensions} & \multicolumn{3}{c|}{Evaluation} \\ \hline
Unaware? & $\theta_a$ & $\theta_{as}$ &Clip? & Extra? & Scaled? & $\theta_{\varepsilon}$ & $\theta_V$ & $\theta_v$ & LL & Accuracy & MS \\ \hline
0 & N/A & N/A & 0 & 0 & 0 & 0.094 & 0.494 & N/A & -0.588 & 0.739 & 0.680 \\
0 & N/A & N/A & 0 & 1 & 0 & 0.196 & 0.034 & 0.576 & -0.558 & 0.761 & 0.672 \\
0 & N/A & N/A & 0 & 0 & 1 & 1.902 & 0.454 & N/A & -0.555 & 0.735 & 0.744 \\
0 & N/A & N/A & 0 & 1 & 1 & 2.654 & 0.195 & 0.376 & -0.541  & 0.759 & 0.729  \\ \hline
1 & 0.711 & 1.632 & 0 & 0 & 0 & 0.973 & 0.679 & N/A & -0.515 & 0.761 & 0.714 \\
1 & 0.804 & 1.766 & 0 & 1 & 0 & 0.780 & 0.442 & 0.212 & -0.508 & 0.766 & 0.706 \\
1 & 0.888 & 2.814 & 0 & 0 & 1 & 5.557 & 0.662 & N/A & -0.495 & 0.766 & 0.731 \\
1 & 0.983 & 3.134 & 0 & 1 & 1 & 6.415 & 0.396 & 0.291 & -0.487 & 0.772 & 0.724 \\ \hline
1 & 0.711 & 1.594 & 1 & 0 & 0 & 0.267 & 0.810 & N/A & -0.493 & 0.770 & 0.726 \\
1 & 0.736 & 1.584 & 1 & 1 & 0 & 0.311 & 0.566 & 0.194 & -0.491 & 0.773 & 0.723 \\
1 & 0.743 & 1.591 & 1 & 0 & 1 & 3.991 & 0.775 & N/A & -0.496 & 0.773 & 0.728 \\
1 & 0.760 & 1.587 & 1 & 1 & 1 & 4.505 & 0.602 & 0.148 & -0.494 & 0.773 & 0.726 \\  \hline
\end{tabular}
\caption{Estimated parameters and performance in the case with only one coupon and reweighting}\label{tab:result_single_rebalanced}
\end{table}

According to Tables \ref{tab:result_single_uniform} and \ref{tab:result_single_rebalanced}, we have several findings:

\begin{itemize}
\item[1.] In general, models with inattention have much better log-likelihood than their counterparts. They also have slightly better prediction accuracies.
\item[2.] Regularizing value function $V$ by clipping does not lead to significant improvement in the log-likelihood or the accuracy. However, the estimated parameters of the inattention model are more stable across different specifications on the utility model.
\item[3.] Including extra feature $v$ in the value estimation leads to improvements in the log-likelihood and the accuracy, but the improvement diminishes as the model becomes more sophisticated. When both the inattention mechanism and the value function regularization by clipping are presented in the model, the improvement is almost insignificant.
\item[4.] Contrary to our intuition, introducing scaling in the error structure does not lead any significant improvement. Rather, it leads to instability in parameter estimations.
\item[5.] For all models, the estimated $\theta_V$ is smaller than 1. This outcome can possibly be explained by the existence of complementary behavioral mechanisms, such as temporal discounting.
\item[6.] The forecasted aggregate coupon redemption ratio is more stable and closer to the observed one under the unawareness models. However, the difference between models is not very large.
\end{itemize}

Next, we examine the estimation results in general cases to determine whether the above findings can be generalized.

\subsection{General cases}

The coupon selection probability on awareness set in general cases is given in Equation \eqref{eq:choice_prob_attention} as
\begin{equation}\label{eq:choice_prob_attention_next}
    \begin{aligned}
    P(c_0|p_x',C_a) = \ & E_{V|D} I(c_0 = \arg\max_{c\in C_a} \{r(p_x',c) + V(f(C_a,c))\} ), \\
    P(c|p_x,C_a) = \ & E_{V|D} I(c_0 \neq \arg\max_{c\in C_a} \{r(p_x',c) + V(f(C_a,c))\} ) \\
    & \cdot I(c = \arg\max_{c\in C/C_0} \{r(p_x',c) + V(f(C,c))\} ) \\
    \approx \ & E_{V|D} I(c_0 \neq \arg\max_{c\in C_a} \{r(p_x',c) + V(f(C_a,c))\} ) \\
    & \cdot E_{V|D} I(c = \arg\max_{c\in C/C_0} \{r(p_x',c) + V(f(C,c))\} ). \\
    \end{aligned}
\end{equation}

Again, we want to find specifications of $V$ with the approximation $\hat{V}$ from equation \eqref{eq:value_comp} and an error term $\varepsilon_V$ such that the above probabilities have tractable forms. One straightforward solution is to consider
\begin{equation}
    \begin{aligned}
    \varepsilon_V(C) & = \tilde{\varepsilon}_V(C) - \tilde{\varepsilon}_V(C_0), \tilde{\varepsilon}_V \sim Gumbel(0,1/\theta_{\varepsilon}) \ \ i.i.d.; \\
    V(C) & = \theta_V\hat{V}(C) + \varepsilon_V(C),
    \end{aligned}
\end{equation}

\noindent where we subtract $\tilde{\varepsilon}_V(C_0)$ for regularity: $V(C_0) \equiv 0$. Notice that this specification is consistent with the one in the case with only one coupon. Now, the coupon selection probability follows the multinomial logit model
\begin{equation}
    E_{V|D} I(c = \arg\max_{c\in C} \{r(p_x',c) + V(f(C,c))\} ) \propto \exp(\theta_{\varepsilon}[r(p_x',c) + \theta_V\hat{V}(f(C,c))]).
\end{equation}

Next, we examine the possible extensions of this basic specification. Among the three possibilities discussed in cases with only one coupon, only the one to include face value $v$ as an extra feature for estimation $V$ can be extended directly.
Scaling $\varepsilon_V(f(C,c))$ with face value $v$ is trivial in implementation but lacks justification in intuition because errors $\varepsilon_V$ for different coupon groups $c$ need to be scaled differently and the overall error structure is broken. Nevertheless, we include it here for comparison.

Clipping the estimation $V$ is difficult to implement because with value clipping the coupon selection probability does not have an analytic form. Moreover, direct computation by sampling methods is intractable given the magnitude of our dataset.
In this study, we implement an approximation of clipping, which is consistent with the one in cases with only one coupon: clipping only affects the probability of choosing no coupon redemption $c_0$. In particular, if there exists a coupon group $c$ with face value $v \leq p_x'$, we remove $c_0$ from the consideration set:
\begin{equation}
    P(c|p_x',C) \propto
    \begin{cases}
        0 & c=c_0 \ \& \ \exists \langle v,T,n \rangle \in C \ \text{such that} \ p_x' \geq v > 0, \\
        \exp(r(p_x',c) + \hat{V}(f(C,c))) & \text{otherwise}.\\
    \end{cases}
\end{equation}

In addition to these three extensions, we consider a nontrivial modification of the error structure. In the current model, coupons that are very similar but not exactly the same are allocated to different groups. Therefore, behavior given coupon set $\{\langle v,T,1 \rangle, \langle v+\epsilon,T,1 \rangle,c_0\}$ will change abruptly as $\epsilon \to 0$. One way to restore the continuity in behavior is to view every coupon as a unique one and consider independent coupon-level estimation error. That is, for any two coupon $\tilde{c}_1$ and $\tilde{c}_2$ in the same group $c$, the errors in the estimations $V(f(C,c))$ are independent. This modification leads to another multinomial logit model
\begin{equation}
    E_{V|D} I(c = \arg\max_{c\in C} \{r(p_x',c) + V(f(C,c))\} ) \propto n\cdot \exp(\theta_{\varepsilon}[r(p_x',c) + \theta_V\hat{V}(f(C,c))]),
\end{equation}

\noindent where $n$ is the number of coupons in group $c$. This extension is called ``iid''.

Table \ref{tab:result_multi_uniform} summarizes the estimated parameters and the performance of models with different combinations of the extensions above. As the computational complexity of the inattention model scales according to the size of $\mathcal{A}(C)$, we limit our attention on records where $|\mathcal{A}(C)| \leq 64$. Such records constitute more than 90\% of the records in the whole dataset.
As a comparison, we also summarize the estimation results on the dataset where each record has $|\mathcal{A}(C)| \leq 16$ in Table \ref{tab:result_multi_sub_uniform}; by calculating the difference between these two tables we can evaluate the stability in the estimated parameters.
The observed aggregate coupon redemption ratio for Tables \ref{tab:result_multi_uniform} and \ref{tab:result_multi_sub_uniform} are 0.803 and 0.786, respectively.

\begin{table}[ht]
    \small
\centering
\begin{adjustwidth}{-0.75cm}{}
\begin{tabular}{|ccc|ccccccc|ccc|}
\hline
\multicolumn{3}{|c|}{Inattention} & \multicolumn{7}{c|}{Utility Model and Extensions} & \multicolumn{3}{c|}{Evaluation} \\ \hline
Unaware? & $\theta_a$ & $\theta_{as}$ &Clip? & Extra? & Scaled? & iid? & $\theta_{\varepsilon}$ & $\theta_V$ & $\theta_v$ & LL & Accuracy & MS \\ \hline
0 & N/A & N/A & 0 & 0 & 0 & 0 & 0.126 & 0.873 & N/A & -0.823 & 0.747 & 0.808 \\
0 & N/A & N/A & 0 & 1 & 0 & 0 & 0.222 & 0.435 & 0.411 & -0.797 & 0.753 & 0.795 \\
0 & N/A & N/A & 0 & 0 & 1 & 0 & 2.423 & 0.800 & N/A & -0.865 & 0.673 & 0.816 \\
0 & N/A & N/A & 0 & 1 & 1 & 0 & 3.931 & 0.456 & 0.380 & -0.852  & 0.674  & 0.798 \\ \hline
0 & N/A & N/A & 0 & 0 & 0 & 1 & 0.134 & 0.898 & N/A & -0.889 & 0.662  & 0.850\\
0 & N/A & N/A & 0 & 1 & 0 & 1 & 0.226 & 0.482 & 0.361 & -0.866 & 0.671  & 0.846\\
0 & N/A & N/A & 0 & 0 & 1 & 1 & 2.333 & 0.995 & N/A & -0.931 & 0.614  & 0.820\\
0 & N/A & N/A & 0 & 1 & 1 & 1 & 3.902 & 0.555 & 0.396 & -0.918 & 0.620  & 0.803 \\ \hline
1 & -0.289 & 2.689 & 0 & 0 & 0 & 0 & 0.332 & 0.722 & N/A & -0.713 & 0.705 & 0.774 \\
1 & -0.268 & 2.844 & 0 & 1 & 0 & 0 & 0.422 & 0.520 & 0.216 & -0.701 & 0.709 & 0.767 \\
1 & -0.098 & 4.309 & 0 & 0 & 1 & 0 & 4.810 & 0.714 & N/A & -0.808 & 0.628 & 0.784 \\
1 & -0.050 & 5.028 & 0 & 1 & 1 & 0 & 5.314 & 0.620 & 0.118 & -0.806 & 0.629 & 0.780 \\ \hline
1 & -0.393 & 2.358 & 0 & 0 & 0 & 1 & 0.388 & 0.702 & N/A & -0.750 & 0.658 & 0.787 \\
1 & -0.374 & 2.387 & 0 & 1 & 0 & 1 & 0.493 & 0.513 & 0.185 & -0.739 & 0.662 & 0.784 \\
1 & -0.244 & 3.890 & 0 & 0 & 1 & 1 & 5.256 & 0.780 & N/A & -0.860 & 0.587 & 0.785 \\
1 & -0.215 & 4.249 & 0 & 1 & 1 & 1 & 5.772 & 0.688 & 0.107 & -0.859 & 0.587 & 0.782 \\ \hline
1 & -0.473 & 1.625 & 1 & 0 & 0 & 0 & 0.367 & 0.680 & N/A & -0.727 & 0.706 & 0.768 \\
1 & -0.441 & 1.642 & 1 & 1 & 0 & 0 & 0.490 & 0.475 & 0.205 & -0.716 & 0.709 & 0.761 \\
1 & -0.425 & 1.682 & 1 & 0 & 1 & 0 & 5.120 & 0.705 & N/A & -0.833 & 0.626 & 0.767 \\
1 & -0.430 & 1.678 & 1 & 1 & 1 & 0 & 4.925 & 0.740 & 0.039 & -0.833 & 0.627 & 0.755 \\ \hline
1 & -0.470 & 1.585 & 1 & 0 & 0 & 1 & 0.413 & 0.680 & N/A & -0.757 & 0.660 & 0.776 \\
1 & -0.438 & 1.600 & 1 & 1 & 0 & 1 & 0.539 & 0.488 & 0.182 & -0.746 & 0.664 & 0.772 \\
1 & -0.412 & 1.674 & 1 & 0 & 1 & 1 & 5.238 & 0.805 & N/A & -0.873 & 0.587 & 0.769 \\
1 & -0.415 & 1.671 & 1 & 1 & 1 & 1 & 5.110 & 0.830 & 0.025 & -0.873 & 0.587 & 0.762 \\ \hline
\end{tabular}
\caption{Estimated parameters and performance on subset $|\mathcal{A}(C)| \leq 64$}\label{tab:result_multi_uniform}
\end{adjustwidth}
\end{table}

\begin{table}[!t]
    \small
\centering
\begin{tabular}{|ccc|ccccc|ccc|}
\hline
\multicolumn{3}{|c|}{Inattention} & \multicolumn{5}{c|}{Utility Model and Extensions} & \multicolumn{3}{c|}{Evaluation} \\ \hline
Unaware? & $\theta_a$ & $\theta_{as}$ &Clip? & Extra? & $\theta_{\varepsilon}$ & $\theta_V$ & $\theta_v$ & LL & Accuracy & MS \\ \hline
0 & N/A & N/A & 0 & 0 & 0.120 & 0.785 & N/A & -0.774 & 0.748 & 0.806 \\
0 & N/A & N/A & 0 & 1 & 0.213 & 0.379 & 0.423 & -0.748 & 0.756 & 0.787 \\ \hline
1 & -0.175 & 2.830 & 0 & 0 & 0.315 & 0.706 & N/A & -0.677 & 0.699 & 0.765 \\
1 & -0.143 & 3.013 & 0 & 1 & 0.398 & 0.504 & 0.222 & -0.666 & 0.702 & 0.757 \\\hline
1 & -0.364 & 1.690 & 1 & 0 & 0.343 & 0.680 & N/A & -0.690 & 0.699 & 0.757 \\
1 & -0.331 & 1.701 & 1 & 1 & 0.461 & 0.471 & 0.210 & -0.681 & 0.702 & 0.749 \\ \hline
\end{tabular}
\caption{Estimated parameters and performance on subset $|\mathcal{A}(C)| \leq 16$}\label{tab:result_multi_sub_uniform}
\end{table}

Our findings from these estimation results are outlined as follows:

\begin{itemize}
\item[1.] Similar to results in the case with only one coupon, models with inattention have better log-likelihood. However, now their accuracies are poor. Because our model of inattention only increases the relative probability of the default option $c_0$, this observation implies that our inattention model underestimates travelers' ability to remember. This claim is also supported by the underestimation of aggregate coupon redemption ratio from unawareness models.
Compared with results from cases with only one coupon, it seems that the prediction accuracy of our unawareness model drops as the size of the coupon set $C$ increases. This observation will be useful when we attempt to improve our model in the future.
\item[2.] Possibly because of a lack of justification, both scaling $\varepsilon_V$ with $v$ and the approximated clipping of $V$ lead to poor log-likelihood and low accuracy. We also observe no improvement in the stability of parameter estimation from value clipping.
\item[3.] The extension with independent coupon-level estimation errors is entirely ineffective. Recall that with this modification, we try to solve the discontinuity problem by introducing heterogeneities among coupons in a coupon group. In the future, we can go in the opposite direction and consider homogeneities between coupons from different coupon groups.
\item[4.] Including the extra feature $v$ in the value estimation leads to significant improvements in both the log-likelihood and the prediction accuracy. This outcome can possibly be explained by the lack of an accurate inattention model or by travelers' bounded rationality in facing difficult optimization problems.
\item[5.] The estimate of $\theta_V$ is again smaller than 1 for all models, suggesting the existence of complementary behavioral mechanisms.
\item[6.] Finally, estimates in Tables \ref{tab:result_multi_uniform} and \ref{tab:result_multi_sub_uniform} are close to each other but different from those in Table \ref{tab:result_single_uniform}. The difference in the parameters of the inattention model is especially large. This finding implies that, despite consistency in terms of mathematical forms, the models developed in this subsection are not natural extensions of the models in cases with only one coupon.
\end{itemize}

\section{Impact of Unawarenesses on Coupons' Promotional Effects}
\label{sec:sim_exp}

In this section, we evaluate the impact of unawarenesses on coupons' promotional effects via simulation. We first show that why a qualitative analysis of such impact is difficult. Recall that the utility gain from choosing the mobility service under the optimal mode choice policy $\pi_{xj}^*$ is
\begin{equation}\label{eq:util_gain_choice}
     u_{xj} + E_{X'|X}\max_{c\in C}\{r(p_x',c) + V^*(f(C,c)) - V^*(f(C))\}.
\end{equation}

When the traveler is only aware of a subset $C_a$ of available coupons instead of the whole set $C$, the opportunity cost $V^*(f(C_a)) - V^*(f(C_a,c))$ of using a coupon $\tilde{c}$ from group $c$ becomes larger. According to equation \eqref{eq:util_gain_choice}, the traveler is then less likely to use the mobility service.
However, unawarenesses also reduces the rate of coupon redemption and the impact of the coupons can last longer.
Moreover, with the \textit{attention recovery} in the payment stage, the traveler may select a coupon $\tilde{c}$ costlier to the operator than any other in $C_a$.
Since all these factors push the promotional effects of coupons in different directions, the impact of inattention is hard to understand qualitatively.

Next, we simulate the traveler's decision flow with the models developed in Sections \ref{sec:model} and \ref{sec:att_model}, and summarize performance metrics including the total trip quantity $N_{trip}$, the total redeemed coupon value $V_{redeemed}$, and the promotional effect
\begin{equation}
    \rho = \frac{\sum(N_{trip} - N_{trip,0})}{\sum V_{redeemed}},
\end{equation}
\noindent where $N_{trip,0}$ is the baseline total trip quantity.

In the simulation, we simplify the trip demand generation and the trip realization processes to reduce computational burden, by assuming that $\lambda_j \equiv 1$, $\log(p_x') \equiv \log(p_x) \sim N(\mu_{pj},\sigma_{pj}^2)$, $u_{xj} \equiv u_0 \in R$. Under these assumptions, the optimal mode selection policy dictates
\begin{equation}
    P(x'=1|C_a,\lambda_j,P_j) = I(u_0 + \max_{c\in C_a}\{r(p_x,c) +  V^*(f(C_a,c)) - V^*(f(C_a)\}) \geq 0).
\end{equation}

Again, since we do not know $V^*$ exactly, we follow the specifications in Section \ref{sec:results} to use $\hat{V}$ from equation \eqref{eq:value_comp} as an estimate and assume a logistic estimation error. Now we have
\begin{equation}
    P(x'=1|C_a,\lambda_j,P_j) = \sigma(\beta[u_0 + \max_{c\in C_a}\{r(p_x,c) +  \hat{V}(f(C_a,c)) - \hat{V}(f(C_a))\}]),
\end{equation}

\noindent where $\sigma$ is the sigmoid function and $\beta$ is a parameter describing the sensitivity of the service selection probability to coupon values.

For the coupon selection probability $P(c|p_x,C_a)$, we specify the same form of $\pi_c(c|p_x',C_a)$ as the basic specification in Section \ref{sec:results}:
\begin{equation}
    \pi_c(c|p_x',C_a) \propto \exp(\theta_{\varepsilon}[r(p_x',c) + \theta_V\hat{V}(f(C_a,c))]).
\end{equation}

In the experiment, we consider the same setting as in Example 2: fare distribution $\log(p_x) \sim N(3.15,0.75^2)$, and the coupon set $C = \{c_0,\langle 10,30,2 \rangle, \langle 10,15,1 \rangle, \langle 5,20,1 \rangle, \langle 20,5,1 \rangle \}$.
For a comprehensive evaluation, we consider various combinations of the default mode selection rate $\lambda_0 = \sigma(\beta u_0)$ and the coupon value sensitivity $\beta$. For parameters in the inattention and coupon selection models, we choose $\theta_{as}=1.5, \theta_a = -0.5, \theta_{\varepsilon} = 0.5$, which are close to the estimations in Section \ref{sec:results}. We also assume that all coupons are activated; that is, $S_a(c) = 1$ for all $c \in C$.

For evaluation, we simulate each $(\lambda_0,\beta)$ case with $T_{max} = \max_c T = 30$ steps and for 250,000 times.
Table \ref{tab:result_sim} presents the simulation results and shows that models with inattention indeed lead to lower promotional effect $\rho$ than their counterparts; in fact, the reduction in $\rho$ can be as great as 10\%. In Table \ref{tab:result_sim}, the default trip quantity $N_{trip,0}$ can be calculated with $\lambda_0 T_{max}$.

\begin{table}[ht]
    \small
\centering
\begin{tabular}{|c|cc|cc|c|}
\hline
\multirow{2}{*}{Inattention} & \multicolumn{2}{c|}{$N_{trip}$} & \multicolumn{2}{c|}{$V_{redeemed}$} & \multirow{2}{*}{$\rho$} \\ \cline{2-5}
 & mean & std & mean & std &  \\ \hline
\multicolumn{6}{|c|}{$\lambda_0 = 0.05, \beta = 0.01$; $N_{trip,0} = 1.5$} \\ \hline
0 & 1.638 & 1.231 & 17.52 & 13.20 & 0.0079 \\
1 & 1.619 & 1.223 & 16.43 & 12.78 & 0.0072  \\ \hline
\multicolumn{6}{|c|}{$\lambda_0 = 0.05, \beta = 0.05$; $N_{trip,0} = 1.5$} \\ \hline
0 & 2.315 & 1.350 & 25.27 & 14.33 & 0.0322 \\
1 & 2.189 & 1.335 & 22.91 & 13.98 & 0.0301  \\ \hline
\multicolumn{6}{|c|}{$\lambda_0 = 0.2, \beta = 0.01$; $N_{trip,0} = 6$} \\ \hline
0 & 6.152 & 2.166 & 43.45 & 11.90 & 0.0035 \\
1 & 6.119 & 2.172 & 40.28 & 12.78 & 0.0030  \\ \hline
\multicolumn{6}{|c|}{$\lambda_0 = 0.2, \beta = 0.05$; $N_{trip,0} = 6$} \\ \hline
0 & 6.718 & 2.054 & 47.70 & 9.83 & 0.0151 \\
1 & 6.588 & 2.102 & 44.06 & 11.41 & 0.0134  \\ \hline
\end{tabular}
\caption{Simulation results of coupons' promotional effects}\label{tab:result_sim}
\end{table}

\section{Discussion and Conclusion}
\label{sec:conclusion}

In this paper, we proposed an inattention mechanism on unawareness to explain the observed deviation of traveler coupon redemption behavior from utility-maximization and estimation results in Section \ref{sec:results} shows that our model indeed leads to a better fit of the dataset compared with baseline models.
Moreover, our simulation experiment in Section \ref{sec:sim_exp} shows that if such unawareness exists, it can lead to a considerable reduction in the promotional effects of coupons.
Therefore, a service operator should be aware of travelers' unawarenesses and take necessary actions. For example, when distributing coupons, the operator should send notifications to travelers to ensure that they are properly incentivized. Moreover, when a traveler's forgetfulness is unavoidable, the operator should include the probability of coupon unawareness in the design of coupon distribution strategies.


The model developed in this study has several limitations worthy of further exploration. First, we did not obtain consistent parameter estimates of the inattention model between the case with only one coupon and general cases. We speculate that the independent coupon-level inattention mechanism employed here is inappropriate. In addition, our consideration of attention state $S_a$ and transition $f_a$ is restricted to the first activation event $I_a$. In the future, we can include more information in $S_a$, such as the time from the most recent activation of each coupon, and consider more complicated transition dynamics $f_a$, such as the Hawkes process or even recurrent neural networks.

Second, we focused on the impact of unawareness in this study, but travelers may also exhibit deliberate attention. In fact, our model of coupon groups dictates a nested consideration structure similar to the one in the classic nested logit model \cite{Ben-Akiva_1985}. However, as mentioned in Section \ref{sec:results}, this nested structure violates regularities because it imposes strong correlations among coupons in the same group but requires independence of coupons from different groups even when these coupons are very similar. In the future, we can employ existing works on \textit{consideration sets} to develop models that are effective in capturing travelers' perceived homogeneity among coupons.

Third, given the limitations in computational power, we failed to extend the value clipping regularization to general cases, and our simple approximation was shown to be ineffective. Further work in developing tractable estimation algorithms of this model is needed.

Finally, our estimation results show that even after we take the impact of unawareness into account, traveler behavior still deviates from utility-maximization decisions. However, it is questionable whether there are alternative decision mechanisms that are both theoretically justifiable and computationally tractable.

\section*{Acknowledgement}

The author thanks Shenhao Wang and Xiang Song for their insightful comments.

\bibliographystyle{plain}
\bibliography{limit_attention}

\newpage

\setcounter{section}{0}
\renewcommand\thesection{\Alph{section}}

\section{Appendix}

\subsection{Value functions in general cases}

In general situations, assumption \ref{thm:a1} does not hold. In this section, we derive value functions for both discrete time settings with general time unit $t_0$ and continuous time setting. We also show the equivalence between the two as the time unit $t_0$ diminishes to 0.

First, consider unit time trip generation rate $\bar{\lambda}_j$: in discrete time settings with a time unit $t_0$, the trip generation probability in each time step is $\bar{\lambda}_j t_0$, while in continuous time setting, the time gap $\delta_t$ between consecutive trips follows the exponential distribution $\delta_t \sim Exp(1/\bar{\lambda}_j)$.

When assumption \ref{thm:a1} does not hold, we need to consider state transition of coupon sets with respect to a variety of time periods. Therefore, we consider the following generalization of coupon state transition function $f$:

\begin{equation}
    \begin{aligned}
        f(C,c,t) & = f_2(f_1(C,t),c)\\
        f_1(C,t) & = \{f_c(c)|c \in C\} \\
        f_c(\langle v,T,n \rangle,t) & =
        \begin{cases}
            \langle v,T-t,n \rangle & v,n > 0, \ T \geq t \\
            c_0 & \text{otherwise}
        \end{cases} \\
        f_2(C,\langle v,T,n \rangle) & =
        \begin{cases}
            (C / \langle v,T,n \rangle) \cup \langle v,T,n-1 \rangle & \ n > 0 \\
            C / \langle v,T,n \rangle & \text{otherwise}
        \end{cases}
    \end{aligned}
\end{equation}

Because now a coupon can expire before a trip ends, we include the state transition $f_c$ within the formulation of individual trips. However, such a transition depends on the realized trip time and ultimately the travel mode. Here, we use $\mathbf{t}_x$ and $\mathbf{t}_x'$ to represent the vector of estimated and realized trip times, respectively. Because these new variables make the problem more complicated, we introduce the following assumptions for a simple characterization.

\begin{assumption}\label{thm:agg_gen}
    \begin{itemize}
        \item[(a)] The time discount factor $\gamma$ equals to 1 (there is no time discount effect).
        \item[(b)] $\sum_{i\neq 1} P(i) V(f_1(C,t_{xi}')) = (1 - P(1)) V(f_1(C,\bar{t}))$; $\bar{t} \in R^+$ can be interpreted as the expected trip time from alternative modes.
    \end{itemize}
\end{assumption}

Now, the value functions in the discrete time setting with time unit $t_0$ can be given as follows.

\begin{equation}\label{eq:value_pi_gen}
    \begin{aligned}
    V_{t_0}^{\pi}(C) = \ &  (1 - \bar{\lambda}_j t_0) V_{t_0}^{\pi}(f_1(C,t_0)) + \bar{\lambda}_j t_0 E_{X}\{V_{t_0}^{\pi}(f_1(C,\lceil\bar{t}/t_0\rceil \cdot t_0)) + \\
    & \pi_{xj}(p_x,\mathbf{u}_{xj},C) [ u_{xj} - V_{t_0}^{\pi}(f_1(C,\lceil\bar{t}/t_0\rceil \cdot t_0)) + E_{X'|X}V_{c,t_0}^{\pi}(p_x',f_1(C,\lceil t_{x1}'/t_0 \rceil \cdot t_0))] \\
    & - \pi_{xj}(p_x,\mathbf{u}_{xj},C_0)  u_{xj} \}, \\
    V_{c,t_0}^{\pi}(p_x',C) = \ & \sum_{c\in C}\pi_c(c|p_x',C) [r(p_x',c) +  V_{t_0}^{\pi}(f_2(C,c))], \\
    V_{t_0}^{\pi}(C_0) = \ & 0.
    \end{aligned}
\end{equation}

On the other hand, the value functions in the continuous time setting can be given as follows.

\begin{equation}\label{eq:cont_value_pi}
    \begin{aligned}
    \bar{V}^{\pi}(C) = \ &  E_{\delta_t|\bar{\lambda}_j} \{ E_{X}\{\bar{V}^{\pi}(f_1(C,\bar{t} + \delta_t)) + \\
    & \pi_{xj}(p_x,\mathbf{u}_{xj},C) [ u_{xj} - \bar{V}^{\pi}(f_1(C,\bar{t} + \delta_t)) + E_{X'|X}\bar{V}_c^{\pi}(p_x',f_1(C,t_{x1}' + \delta_t))] \\
    & - \pi_{xj}(p_x,\mathbf{u}_{xj},C_0)  u_{xj} \}\}, \\
    \bar{V}_c^{\pi}(p_x',C) = \ & \sum_{c\in C}\pi_c(c|p_x',C) [r(p_x',c) +  \bar{V}^{\pi}(f_2(C,c))], \\
    \bar{V}^{\pi}(C_0) = \ & 0.
    \end{aligned}
\end{equation}

The similarity between the above two equations actually leads to the following equivalence result:

\begin{proposition}\label{thm:p3}
  The value function $V^{\pi}_{t_0}$ in the discrete time setting converges to the value function $\bar{V}^{\pi}$ in the continuous time setting as the time step size $t_0$ diminishes to 0:
  \begin{equation}\label{eq:equal_value_pi}
    \bar{V}^{\pi}(C) = \lim_{t_0\to 0}V^{\pi}_{t_0}(C), \ \forall C \in \mathcal{C}.
  \end{equation}
\end{proposition}

\begin{proof}[Proof of Proposition \ref{thm:p3}]

    First, let $\bar{Q}^{\pi}(C,\delta_t)$ be
    \begin{equation}
        \begin{aligned}
            & E_{X}\{\bar{V}^{\pi}(f_1(C,\bar{t} + \delta_t)) + \\
            & \pi_{xj}(p_x,\mathbf{u}_{xj},C) [ u_{xj} - \bar{V}^{\pi}(f_1(C,\bar{t} + \delta_t)) + E_{X'|X}\bar{V}_c^{\pi}(p_x',f_1(C,t_{x1}' + \delta_t))] \\
            & - \pi_{xj}(p_x,\mathbf{u}_{xj},C_0)  u_{xj}\}.
        \end{aligned}
    \end{equation}

    Then we have
    \begin{equation}
        \begin{aligned}
        - \frac{d \{ \bar{V}^{\pi}(f_1(C,t)) \} }{dt} = \ & - \frac{d}{dt} \int_0^{\infty} \bar{\lambda}_j e^{-\bar{\lambda}_j\delta_t} \bar{Q}^{\pi}(C,t + \delta_t) d \delta_t \\
        = \ & - \frac{d}{dt} \int_t^{\infty} \bar{\lambda}_j e^{-\bar{\lambda}_j(\delta_t - t)} \bar{Q}^{\pi}(C,\delta_t) d \delta_t \\
        = \ & \bar{\lambda}_j \bar{Q}^{\pi}(C,t) - \int_t^{\infty} \bar{\lambda}_j^2 e^{-\bar{\lambda}_j(\delta_t - t)}  \bar{Q}^{\pi}(C,\delta_t) d \delta_t \\
        = \ & \bar{\lambda}_j [ \bar{Q}^{\pi}(C,t) - \bar{V}^{\pi}(f_1(C,t))].
        \end{aligned}
    \end{equation}

    If we let $t=0$ in the above equations, we have
    \begin{equation}\label{eq:diff_cont}
        - \frac{d \{ \bar{V}^{\pi}(f_1(C,t)) \} }{dt}|_{t=0} = \bar{\lambda}_j [ \bar{Q}^{\pi}(C,0) - \bar{V}^{\pi}(C)].
    \end{equation}

    Secondly, if we let $Q_{t_0}^{\pi}(C,\delta_t)$ be
    \begin{equation}
        \begin{aligned}
            & E_{X}\{V_{t_0}^{\pi}(f_1(C,\delta_t+\lceil\bar{t}/t_0\rceil \cdot t_0)) + \\
            & \pi_{xj}(p_x,\mathbf{u}_{xj},C) [ u_{xj} - V_{t_0}^{\pi}(f_1(C,\delta_t+\lceil\bar{t}/t_0\rceil \cdot t_0)) + E_{X'|X}V_{c,t_0}^{\pi}(p_x',f_1(C,\delta_t+\lceil t_{x1}'/t_0 \rceil \cdot t_0))] \\
            & - \pi_{xj}(p_x,\mathbf{u}_{xj},C_0)  u_{xj} \},
        \end{aligned}
    \end{equation}

    \noindent we have

    \begin{equation}\label{eq:diff_discrete}
        \begin{aligned}
        - \frac{d \{ \lim_{t_0\to 0} V^{\pi}_{t_0}(f_1(C,t)) \} }{dt}|_{t=0} = \ & \lim_{t_0\to 0} \frac{V^{\pi}_{t_0}(C) - V^{\pi}_{t_0}(f_1(C,t_0))}{t_0} \\
        = \ & \lim_{t_0\to 0} \bar{\lambda}_j(Q_{t_0}^{\pi}(C,0) - V_{t_0}^{\pi}(f_1(C,t_0))) \\
        = \ & \bar{\lambda}_j(\lim_{t_0\to 0} Q_{t_0}^{\pi}(C,0) - \lim_{t_0\to 0} V_{t_0}^{\pi}((C)).
        \end{aligned}
    \end{equation}

    Because

    \begin{equation}
        \begin{aligned}
        \lim_{t_0\to 0} Q_{t_0}^{\pi}(C,0) = \ &  E_{X}\{\lim_{t_0\to 0}V_{t_0}^{\pi}(f_1(C,\bar{t})) + \\
        & \pi_{xj}(p_x,\mathbf{u}_{xj},C) [ u_{xj} - \lim_{t_0\to 0}V_{t_0}^{\pi}(f_1(C,\bar{t})) + E_{X'|X}\lim_{t_0\to 0}V_{c,t_0}^{\pi}(p_x',f_1(C,t_{x1}'))] \\
        & - \pi_{xj}(p_x,\mathbf{u}_{xj},C_0)  u_{xj} \},
        \end{aligned}
    \end{equation}

    \noindent we can see that equations \eqref{eq:diff_cont} and \eqref{eq:diff_discrete} actually refer to the same differential equation. Given boundary condition $\lim_{t_0\to 0} V^{\pi}_{t_0}(C_0) = \bar{V}^{\pi}(C_0) = 0$, we know the solution to the differential equation is unique and therefore $\lim_{t_0\to 0} V^{\pi}_{t_0}(C) = \bar{V}^{\pi}(C)$.

\end{proof}

\subsection{Proofs}

\subsubsection{Proof of Corollary \ref{thm:c1}}

\begin{proof}[Proof of Corollary \ref{thm:c1}]

By replacing the coupon set $C$ in equations \eqref{eq:util_ex_ante}, \eqref{eq:util_interim} and \eqref{eq:util_ex_post} with the default set $C_0$, we have
\begin{equation}
    \begin{aligned}
        U^{\pi}(C_0) & = (1 - \lambda_j) \gamma U^{\pi}(C_0) + \lambda_j E_{X} U_{xj}^{\pi}(p_x,\mathbf{u}_{xj},C_0), \\
        U_{xj}^{\pi}(p_x,\mathbf{u}_{xj},C_0) & = (1 - \pi_{xj}(p_x,\mathbf{u}_{xj},C_0)) \gamma U^{\pi}(C_0) + \pi_{xj}(p_x,\mathbf{u}_{xj},C_0) [ u_{xj} + E_{X'|X}U_c^{\pi}(p_x',C_0)] + u_{x\tilde{1}j}, \\
        U_c^{\pi}(p_x',C_0) & = \gamma U^{\pi}(C_0),
    \end{aligned}
\end{equation}

\noindent which we can simplify to
\begin{equation}
    \begin{aligned}
        U^{\pi}(C_0) & = (1 - \lambda_j) \gamma U^{\pi}(C_0) + \lambda_j E_{X} \{ (1 - \pi_{xj}(p_x,\mathbf{u}_{xj},C_0)) \gamma U^{\pi}(C_0) + \pi_{xj}(p_x,\mathbf{u}_{xj},C_0) [ u_{xj} + \gamma U^{\pi}(C_0)] + u_{x\tilde{1}j} \}\\
        & = \gamma U^{\pi}(C_0) + \lambda_j E_{X}[u_{x\tilde{1}j} + \pi_{xj}(p_x,\mathbf{u}_{xj},C_0)  u_{xj}] \\
        & = \frac{1}{1 - \gamma}\lambda_j E_{X}[u_{x\tilde{1}j} + \pi_{xj}(p_x,\mathbf{u}_{xj},C_0)  u_{xj}].
    \end{aligned}
\end{equation}

\end{proof}

\subsubsection{Proof of Proposition \ref{thm:p1}}

\begin{proof}[Proof of Proposition \ref{thm:p1}]

First, it is easy to see that $U_c^{\pi}(p_x',C_0) = \gamma U^{\pi}(f(C_0)) = \gamma U^{\pi}(C_0)$.
Now, by the definition of $V^{\pi}$ and $V_c^{\pi}$, and equations \eqref{eq:util_ex_ante}, \eqref{eq:util_interim} and \eqref{eq:util_ex_post}, we have
\begin{equation}
    \begin{aligned}
    V_c^{\pi}(p_x',C) = \ & \sum_{c\in C}\pi_c(c|p_x',C) [r(p_x',c) + \gamma U^{\pi}(f(C,c))] - \gamma U^{\pi}(f(C_0)) \\
    = \ & \sum_{c\in C}\pi_c(c|p_x',C) [r(p_x',c) + \gamma V^{\pi}(f(C,c))],
    \end{aligned}
\end{equation}

\noindent and
\begin{equation}
    \begin{aligned}
    V^{\pi}(C) = \ &  \gamma [U^{\pi}(f(C)) - U^{\pi}(f(C_0))] + \lambda_j E_{X}\{ \\
    & \pi_{xj}(p_x,\mathbf{u}_{xj},C) [ u_{xj} + E_{X'|X}U_c^{\pi}(p_x',C) - \gamma  U^{\pi}(f(C))] \\
    & - \pi_{xj}(p_x,\mathbf{u}_{xj},C_0) [ u_{xj} + E_{X'|X} U_c^{\pi}(p_x',C_0) - \gamma  U^{\pi}(f(C_0))] \} \\
    = \ & \gamma V^{\pi}(f(C)) + \lambda_j E_{X}\{ \\
    & \pi_{xj}(p_x,\mathbf{u}_{xj},C) [ u_{xj} + E_{X'|X}V_c^{\pi}(p_x',C) - \gamma  V^{\pi}(f(C)) ] - \pi_{xj}(p_x,\mathbf{u}_{xj},C_0)  u_{xj} \}. \\
    \end{aligned}
\end{equation}

\end{proof}

\subsubsection{Proof of Proposition \ref{thm:p2}}

\begin{lemma}\label{thm:l1}
  If $x \geq y$, then we have
  \begin{equation}
    I(x \geq 0) (x - y) \geq \max(x,0) - \max(y,0) \geq I(y \geq 0) (x - y),
  \end{equation}
  \noindent where $I(\cdot)$ is the indicator function.
\end{lemma}

\begin{proof}[Proof of Lemma \ref{thm:l1}]
    Since $x \geq y$, we have $y \geq 0 \Rightarrow x \geq 0$, and $x < 0 \Rightarrow y < 0$. Therefore
    \begin{equation}
        \begin{aligned}
            \max(x,0) - \max(y,0) = \ & I(y \geq 0) (\max(x,0) - \max(y,0)) \\
            & + I(y < 0) (\max(x,0) - \max(y,0)) \\
            = \ & I(y \geq 0) (x-y) + I(g(x) < 0) \max(x,0) \\
            \geq \ & I(y \geq 0) (x-y).
        \end{aligned}
    \end{equation}

    Similarly, we have
    \begin{equation}
        \begin{aligned}
            \max(x,0) - \max(y,0) = \ & I(x \geq 0) (\max(x,0) - \max(y,0)) \\
            & + I(x < 0) (\max(x,0) - \max(y,0)) \\
            = \ & I(x \geq 0) (x - \max(y,0)) \\
            \leq \ & I(x \geq 0) (x-y).
        \end{aligned}
    \end{equation}

\end{proof}

\begin{proof}[Proof of Proposition \ref{thm:p2}]

First, recall that in the coupon selection stage, default action $c_0$ is always available. Therefore the \textit{ex post} value function $V_c^*(p_x',C)$ has a lower bound
\begin{equation}
    V_c^*(p_x',C) \geq r(p_x',c_0) +  V^*(f(C,c_0)) =  V^*(f(C)).
\end{equation}

This further leads to
\begin{equation}
    u_{xj} + E_{X'|X}V_c^*(p_x',C) -  V^*(f(C)) \geq u_{xj}.
\end{equation}

Now, applying Lemma \ref{thm:l1} with $x = u_{xj} + E_{X'|X}V_c^*(p_x',C) -  V^*(f(C))$ and $y = u_{xj}$, we have
\begin{equation}
    \begin{aligned}
        \max(u_{xj} + E_{X'|X}V_c^*(p_x',C) -  V^*(f(C)),0) - \max(u_{xj},0) & \geq I(u_{xj} \geq 0)[E_{X'|X}V_c^*(p_x',C) -  V^*(f(C))], \\
        \max(u_{xj} + E_{X'|X}V_c^*(p_x',C) -  V^*(f(C)),0) - \max(u_{xj},0) & \leq I(u_{xj} + E_{X'|X}V_c^*(p_x',C) -  V^*(f(C)) \geq 0)\cdot\\
        & \text{\space\space\space\space\space} [E_{X'|X}V_c^*(p_x',C) -  V^*(f(C))].
    \end{aligned}
\end{equation}

Putting above inequalities back to the \textit{ex ante} value function in equation \eqref{eq:value_opt}, we have
\begin{equation}
    \begin{aligned}
        V^*(C) & \geq  V^*(f(C)) + \lambda_j E_{X}I(u_{xj} \geq 0)[E_{X'|X}V_c^*(p_x',C) -  V^*(f(C))], \\
        V^*(C) & \leq  V^*(f(C)) + \lambda_j E_{X}I(u_{xj} + E_{X'|X}V_c^*(p_x',C) -  V^*(f(C)) \geq 0)[E_{X'|X}V_c^*(p_x',C) -  V^*(f(C))].
    \end{aligned}
\end{equation}

Since $\lambda_j E_{X}I(u_{xj} + E_{X'|X}V_c^*(p_x',C) -  V^*(f(C)) \geq 0)$ is exactly the service selection probability $P(x'=1|C,\lambda_j,P_j)$, the equations above can be further simplified to
\begin{equation}
    \begin{aligned}
        V^*(C) & \geq  V^*(f(C)) + E_{x',p_x'|C_0}I(x'=1)[V_c^*(p_x',C) -  V^*(f(C))], \\
        V^*(C) & \leq  V^*(f(C)) + E_{x',p_x'|C}I(x'=1)[V_c^*(p_x',C) -  V^*(f(C))],
    \end{aligned}
\end{equation}

Notice that in equations above we hide the fact that $x'$ and $p_x'$ are conditional on $\lambda_j$, $P_j$ and $P_j'$.

Now, given boundary condition $V^*(C_0) = V^L(C_0)$, we can show by the induction principle that for every $C \in \mathcal{C}$
\begin{equation}
    \begin{aligned}
        V^*(C) & \geq  V^*(f(C)) + E_{x',p_x'|C_0}I(x'=1)[V_c^*(p_x',C) -  V^*(f(C))] \\
        & = E_{x',p_x'|C_0}I(x'=0)  V^*(f(C)) + E_{x',p_x'|C_0}I(x'=1)V_c^*(p_x',C) \\
        & \geq E_{x',p_x'|C_0}I(x'=0)  V^L(f(C)) + E_{x',p_x'|C_0}I(x'=1)V_c^{L}(p_x',C)\\
        & =  V^L(f(C)) + E_{x',p_x'|C_0}I(x'=1)[V_c^{L}(p_x',C) -  V^L(f(C))] \\
        & = V^L(C).
    \end{aligned}
\end{equation}

Similarly we have $V^*(C) \leq V^U(C), \forall \ C \in \mathcal{C}$ given $V^*(C_0) = V^U(C_0)$.

\end{proof}

\newpage

\subsection{Notation table}

\begin{table}[!ht]
    \small
\centering
\begin{tabular}{cc}
\hline
\multicolumn{2}{c}{coupon related} \\
\hline
face value & $v$ \\
time to expire & $T$ \\
number of coupons in the group & $n$ \\
coupon & $\tilde{c}$ \\
coupon group & $c$ \\
default (zero-valued) coupon group & $c_0$ \\
coupon set & $C$ \\
awareness coupon subset & $C_a$ \\
set of all coupon sets & $\mathcal{C}$ \\
set of all awareness subset of set $C$ & $\mathcal{A}(C)$ \\
\hline
\multicolumn{2}{c}{model related} \\
\hline
trip demand generation rate of traveler $j$ & $\lambda_j$ \\
mean of log fare $\log(p_x')$ of traveler $j$ & $\mu_{pj}$ \\
standard deviation of log fare $\log(p_x')$ of traveler $j$ & $\sigma_{pj}$ \\
estimated trip fare & $p_x$ \\
vector of trip utilities from different travel modes & $\mathbf{u}_{xj}$ \\
relative utility gain in taking the target mobility service & $u_{xj}$ \\
mode selection of travelers & $i$ \\
whether there is a trip served by the target mobility service & $x'$ \\
realized trip fare & $p_x'$ \\
state of attention & $S_a$ \\
coupon activation record function & $I_a(\cdot)$ \\
coupon redemption utility function & $r(\cdot,\cdot)$ \\
state transition functions of coupon set & $f(\cdot,\cdot)$ \\
state transition functions of coupon & $f_c(\cdot)$ \\
state transition function of attention & $f_a(\cdot,\cdot,\cdot)$ \\
discount factor & $\gamma$ \\
\hline
\multicolumn{2}{c}{value function related} \\
\hline
mode selection policy of traveler $j$ & $\pi_{xj}(\cdot,\cdot,\cdot)$ \\
coupon selection policy of travelers & $\pi_c(\cdot,\cdot)$ \\
expected accumulated utility under policy $\pi$ & $U^{\pi}(\cdot)$ \\
utility gain from coupon sets under policy $\pi$ & $V^{\pi}(\cdot)$ \\
optimal utility gain from coupon sets & $V^{*}(\cdot)$ \\
lower \& upper bound of $V^*$ & $V^L(\cdot),V^U(\cdot)$ \\
approximated utility gain from coupon sets & $\hat{V}(\cdot)$ \\
estimation of utility gain from coupon sets (random variable) & $V(\cdot)$ \\
estimation error (random variable) & $\varepsilon_V(\cdot)$ \\
\hline
\multicolumn{2}{c}{others} \\
\hline
dataset & $D$ \\
model parameters & $\theta$ \\
sigmoid function & $\sigma(\cdot)$ \\
indicator function & $I(\cdot)$ \\
\hline
\end{tabular}
\caption{Summary of major notations}\label{tab:var}
\end{table}

\newpage

\subsection{Estimation results with the whole dataset}

\begin{table}[ht]
    \small
\centering
\begin{tabular}{|ccc|cccccc|cc|}
\hline
\multicolumn{3}{|c|}{Inattention} & \multicolumn{6}{c|}{Utility Model and Extensions} & \multicolumn{2}{c|}{Evaluation} \\ \hline
Unaware? & $\theta_a$ & $\theta_{as}$ &Clip? & Extra? & Scaled? & $\theta_{\varepsilon}$ & $\theta_V$ & $\theta_v$ & LL & Accuracy \\ \hline
0 & N/A & N/A & 0 & 0 & 0 & 0.064 & 0.852 & N/A & -0.638 & 0.677 \\
0 & N/A & N/A & 0 & 1 & 0 & 0.177 & 0.114 & 0.669 & -0.594 & 0.708 \\
0 & N/A & N/A & 0 & 0 & 1 & 1.144 & 0.665 & N/A & -0.639 & 0.663 \\
0 & N/A & N/A & 0 & 1 & 1 & 3.676 & 0.061 & 0.712 & -0.593 & 0.706  \\ \hline
1 & 0.392 & 2.444 & 0 & 0 & 0 & 0.379 & 0.785 & N/A & -0.560 & 0.701 \\
1 & 0.613 & 3.100 & 0 & 1 & 0 & 0.361 & 0.440 & 0.327 & -0.547 & 0.716 \\
1 & 0.385 & 2.600 & 0 & 0 & 1 & 6.135 & 0.760 & N/A & -0.555 & 0.701 \\
1 & 0.594 & 3.841 & 0 & 1 & 1 & 6.429 & 0.395 & 0.365 & -0.542 & 0.716 \\ \hline
1 & 0.377 & 1.722 & 1 & 0 & 0 & 0.217 & 0.913 & N/A & -0.549 & 0.714 \\
1 & 0.438 & 1.684 & 1 & 1 & 0 & 0.293 & 0.534 & 0.289 & -0.543 & 0.718 \\
1 & 0.370 & 1.732 & 1 & 0 & 1 & 4.814 & 0.871 & N/A & -0.549 & 0.714 \\
1 & 0.437 & 1.691 & 1 & 1 & 1 & 6.479 & 0.505 & 0.289 & -0.543 & 0.718 \\ \hline
\end{tabular}
\caption{Estimated parameters and performance in the case with only one coupon on the whole dataset}\label{tab:full_result_single_uniform}
\end{table}

\begin{table}[ht]
    \small
\centering
\begin{tabular}{|ccc|ccccc|cc|}
\hline
\multicolumn{3}{|c|}{Inattention} & \multicolumn{5}{c|}{Utility Model and Extensions} & \multicolumn{2}{c|}{Evaluation} \\ \hline
Unaware? & $\theta_a$ & $\theta_{as}$ &Clip? & Extra? & $\theta_{\varepsilon}$ & $\theta_V$ & $\theta_v$ & LL & Accuracy \\ \hline
0 & N/A & N/A & 0 & 0 & 0.124 & 0.802 & N/A & -0.832 & 0.709 \\
0 & N/A & N/A & 0 & 1 & 0.242 & 0.329 & 0.466 & -0.795 & 0.725 \\ \hline
1 & -0.490 & 2.530 & 0 & 0 & 0.346 & 0.724 & N/A & -0.709 & 0.699 \\
1 & -0.448 & 2.723 & 0 & 1 & 0.447 & 0.490 & 0.245 & -0.693 & 0.705 \\\hline
1 & -0.700 & 1.561 & 1 & 0 & 0.377 & 0.708 & N/A & -0.722 & 0.697 \\
1 & -0.664 & 1.577 & 1 & 1 & 0.504 & 0.482 & 0.219 & -0.710 & 0.703 \\ \hline
\end{tabular}
\caption{Estimated parameters and performance on subset $|\mathcal{A}(C)| \leq 64$ of the whole dataset}\label{tab:full_result_multi_uniform}
\end{table}

\end{document}